%% file: main.tex
\documentclass[11pt]{article}
\usepackage[dvipsnames]{xcolor}
\usepackage{graphicx}
\usepackage{fullpage}
\usepackage{amsmath,amsthm,amssymb}
\usepackage[hyphens]{url}
\usepackage{hyperref}
\usepackage[nameinlink,capitalise]{cleveref}
\usepackage{nameref}
\usepackage{tikz}
\usepackage{mlmodern} %
\usepackage[many]{tcolorbox}
\usepackage{enumitem}
\usepackage{dsfont}
\usepackage{nccmath}
\usepackage{setspace}

\newif\ifshowproofs
\showproofstrue

\AtBeginEnvironment{bmatrix}{\setlength{\arraycolsep}{3.5pt}}

\definecolor{darkgreen}{rgb}{0.0, 0.5, 0.0}
\def\comments{1}

\ifnum\comments=1
\newcommand{\pras}[1]{\textcolor{BrickRed}{\sf{[#1 --PR]}}}
\newcommand{\kangning}[1]{\textcolor{orange}{\sf{[#1 --KW]}}}
\newcommand{\zihan}[1]{\textcolor{purple}{\sf{[#1 --ZT]}}}
\newcommand{\moses}[1]{\textcolor{darkgreen}{\sf{[#1 --MC]}}}
\newcommand{\edit}[2]{\textcolor{red}{#1}}
\fi

\ifnum\comments=0
\newcommand{\pras}[1]{}
\newcommand{\kangning}[1]{}
\newcommand{\moses}[1]{}
\newcommand{\zihan}[1]{}
\newcommand{\edit}[2]{#1}
\fi

\title{Metric Distortion for Tournament Voting and Beyond}

\author{%
  \begin{tabular}{@{}c@{\hspace{5em}}c@{}}
    Moses Charikar       & Prasanna Ramakrishnan \\  
    {\sl Stanford University} & {\sl Stanford University} \\
    \texttt{moses@cs.stanford.edu} & \texttt{pras1712@stanford.edu} \\[1.5em]
    Zihan Tan            & Kangning Wang         \\
    {\sl Rutgers University}  & {\sl Rutgers University} \\
    \texttt{zihantan1993@gmail.com} & \texttt{kn.w@rutgers.edu}
  \end{tabular}%
}

\date{}

\input{notations}

\begin{document}

\maketitle

\input{abstract}

\thispagestyle{empty}
\newpage

\renewcommand{\contentsname}{Contents}
\doublespacing
\tableofcontents
\singlespacing

\thispagestyle{empty}
\newpage
\setcounter{page}{1}

\input{sections/introduction}

\input{sections/our}

\input{sections/related}

\input{sections/preliminaries}

\input{sections/biased}

\input{sections/determ_two}

\input{sections/determ_two_ub}

\input{sections/determ_two_lb}

\input{sections/k-selection}

\input{sections/stable-lotteries}

\input{sections/main_lemma}

\input{sections/determ_any}

\input{sections/random_any}

\input{sections/conclusions}

\section*{Acknowledgments}
Moses Charikar is supported by a Simons Investigator Award. Prasanna Ramakrishnan is supported by Moses Charikar’s Simons Investigator Award and Li-Yang Tan’s NSF awards 1942123, 2211237, 2224246, Sloan Research Fellowship, and Google Research Scholar Award.

\bibliographystyle{alpha}
\bibliography{ref}

\appendix

\input{sections/omitted}

\end{document}

%% file: notations.tex
\newtheorem{theorem}{Theorem}[section]
\newtheorem{lemma}[theorem]{Lemma}

\newtheorem{corollary}[theorem]{Corollary}

\newtheorem{observation}[theorem]{Observation}

\theoremstyle{definition}
\newtheorem{definition}[theorem]{Definition}
\newtheorem{example}[theorem]{Example}
\newtheorem{remark}{Remark} %

\DeclareMathOperator*{\E}{\mathbb{E}}
\renewcommand{\d}{\,\mathrm{d}}

\crefname{Program}{Program}{Programs}
\creflabelformat{Program}{(#2\textup{#1})#3}

\makeatletter
\def\@fnsymbol#1{\ensuremath{\ifcase#1\or \dagger\or \ddagger\or
   \mathsection\or \mathparagraph\or \|\or **\or \dagger\dagger
   \or \ddagger\ddagger \else\@ctrerr\fi}}
\makeatother

\definecolor{linkc}{rgb}{0.4, 0.3, 0.7}
\definecolor{citec}{rgb}{0.4, 0.7, 0.7}
\definecolor{urlc}{rgb}{0.5, 0.1, 0.2}
\hypersetup{
    colorlinks=true,
    linkcolor=linkc,
    citecolor=citec,
    urlcolor=urlc
}

\newtcolorbox[auto counter,number within=section]{mybox}[2][]{colback=orange!8!white,colframe=gray!30!black,title={#2},#1}

\newcommand*\dif{\mathop{}\!\mathrm{d}}
\newcommand{\cg}{\succ}

\newcommand{\cgeq}{\succeq}

\newcommand{\oneif}{\mathds{1}}

\newcommand{\score}{{\mathsf{score}}}

\DeclareMathOperator{\plu}{plu}
\DeclareMathOperator{\SC}{SC}

%% file: abstract.tex
\begin{abstract}
In the well-studied metric distortion problem in social choice, we have voters and candidates located in a shared metric space, and the objective is to design a voting rule that selects a candidate with minimal total distance to the voters. However, the voting rule has limited information about the distances in the metric, such as each voter's ordinal rankings of the candidates in order of distances. The central question is whether we can design rules that, for any election and underlying metric space, select a candidate whose total cost deviates from the optimal by only a small factor, referred to as the \emph{distortion}.

A long line of work resolved the optimal distortion of \emph{deterministic rules}, and recent work resolved the optimal distortion of randomized (weighted) \emph{tournament rules}, which only use the aggregate preferences between pairs of candidates. In both cases, simple rules achieve the optimal distortion of 3. Can we achieve the best of both worlds: a \emph{deterministic tournament rule} matching the lower bound of $3$? Prior to our work, the best rules have distortion $2 + \sqrt{5} \approx 4.2361$.

In this work, we establish a lower bound of $3.1128$ on the distortion of any deterministic tournament rule, even when there are only 5 candidates, and improve the upper bound with a novel rule guaranteeing distortion $3.9312$. We then generalize tournament rules to the class of \emph{$k$-tournament rules} which obtain the aggregate preferences between $k$-tuples of candidates. We show that there is a family of deterministic $k$-tournament rules that achieves distortion approaching $3$ as $k$ grows. Finally, we show that even with $k = 3$, a randomized $k$-tournament rule can achieve distortion less than $3$, which had been a longstanding barrier even for the larger class of ranked voting rules.
\end{abstract}

%% file: sections/introduction.tex
\section{Introduction}

Social choice theory studies methods of aggregating individual preferences to reach a collective decision. This process involves two main steps: first, eliciting individual preferences, and second, applying a social choice rule to decide the winner based on these elicited preferences.

In preference elicitation, much of social choice theory considers ranked ballots, where each voter reports a ranking of all the candidates according to her preference. However, ranked voting is not the only historically significant format. Back in the 13th century, Ramon Llull \cite{llull1274artifitium,hagele2000lulls} introduced a voting rule that holds a mini-election between each pair of candidates and declares as the winner the candidate with the most pairwise wins. This rule is now known as the Copeland rule. Interestingly, if instead of counting pairwise wins it chose the candidate with the most total votes during the process, the rule would become the Borda rule, first proposed by Nicholas of Cusa in the 15th century \cite{cusa1433,christianson2008church}.\footnote{Some sources attribute the Borda rule to Llull as well, but according to \cite{colomer2013ramon}, ``[i]n contrast to some previous tentative suggestions, careful reading of Llull's papers demonstrates that he did not propose a rank-order count system, such as those proposed later on by Cusanus and Borda.''}

The Copeland rule and the Borda rule both fall within the category of \emph{tournament rules}, meaning the only information they elicit is how many voters prefer $a$ to $b$ (written as $a \succ b$) for each pair of candidates $a$ and $b$.\footnote{In the literature, the term ``tournament rules'' more often refers to \emph{unweighted} tournament rules, where the only information they elicit is whether or not a majority of voters prefer $a$ to $b$ for each pair of candidates $a$ and $b$. Since our focus is on weighted tournament rules that also know the margins of victory, we use the term to refer to weighted ones throughout our discussion similar to some prior works.} Researchers have been studying a large variety of tournament rules; we discuss some of these in \cref{sec:related} and refer the reader to the nice surveys of \cite{brandt2016tournament,fischer2016weighted} for a comprehensive overview.

Several factors make tournament rules compelling for theoretical study and practical use. %
Given our strong understanding of voting when there are only two candidates, tournament rules are a way to use the tools of graph theory to lift this understanding up to settings with more than two candidates. As a result, tournament rules, like the Copeland and Borda rules, tend to feel natural and intuitive. \cite{fishburn1977condorcet} also notes that from a data perspective, tournament rules are more efficient than ranked voting rules, since they require only information on the quadratically many comparisons between pairs of candidates, rather than the exponentially many possible rankings. %
In this sense, the tournament graph is a compressed representation of an election, which makes it easier to summarize, evaluate, and audit the results.  Similarly, tournament rules are especially appealing when eliciting a full ranking of candidates from each voter becomes cumbersome, since an accurate tournament graph can be computed by querying random voters about their preferences between pairs of candidates. Recently, this model of preference elicitation has become widespread in \emph{reinforcement learning from human feedback} (RLHF), where generative models are fine-tuned using human preferences over pairs of outputs \cite{christiano2017deep,DBLP:conf/nips/Ouyang0JAWMZASR22,touvron2023llama,zhu2023principled}.

\edit{}{The appeal of tournament rules naturally prompts the question: what can we do by incorporating information just beyond pairwise comparisons? 
With this in mind, we introduce \emph{$k$-tournament rules} as a generalization of tournament rules. These rules elicit information based on $k$-wise comparisons: for each tuple of $k$ candidates $(c_1, c_2, \ldots, c_k)$, they ask how many voters have the preference $c_1 \succ c_2 \succ \cdots \succ c_k$. $k$-tournaments capture preference elicitation methods in some RLHF settings: for example, in the work of \cite{DBLP:conf/nips/Ouyang0JAWMZASR22}, they ``present labelers with anywhere between $k = 4$ and $k = 9$ responses to rank.'' Though this generalization feels natural and has seen practical use, there is little published discussion on these rules, possibly due to the historical absence of a framework that can guide the design and evaluation of $k$-tournaments.}

\paragraph{The metric distortion framework.} In this work, we study tournaments \edit{}{and $k$-tournaments} through the lens of metric distortion \cite{DBLP:conf/aaai/AnshelevichBP15,DBLP:journals/ai/AnshelevichBEPS18}. The metric distortion framework is inspired by the proximity spatial model of voting \cite{enelow1984spatial,enelow1990advances,merrill1999unified,armstrong2020analyzing} which assumes that voters and candidates are in the same metric space, where each voter prefers those candidates who are closer to her. The goal is to select a candidate with a small social cost, defined as the sum of distances from the candidate to all voters. If a voting rule (that only elicits partial information) always has social cost at most $t$ times the optimal social cost regardless of the underlying metric space, then we say this rule has metric distortion at most $t$.

Early work on metric distortion \cite{DBLP:conf/aaai/AnshelevichBP15, DBLP:journals/ai/AnshelevichBEPS18,DBLP:conf/sigecom/GoelKM17} found it to be a useful tool for gaining insight into existing voting rules by using metric spaces to probe their strengths and weaknesses. More recent work \cite{DBLP:conf/ec/MunagalaW19,DBLP:conf/aaai/Kempe20a,DBLP:conf/focs/GkatzelisHS20,DBLP:conf/ijcai/KizilkayaK22,DBLP:conf/sigecom/Kizilkaya023, DBLP:journals/jacm/CharikarRWW24} has found another alluring appeal of metric distortion: to use the goal of optimal distortion to guide the discovery of novel voting rules. While it would be reasonable to expect these rules to be specialized to the metric setting (like some rules proposed by \cite{DBLP:conf/aaai/Kempe20a} and \cite{DBLP:conf/soda/CharikarR22}, which encode the metrics as constraints in a linear program), in a pleasant surprise, they have largely been simple and intuitive, without any need for metric spaces to describe them.

\edit{}{Just as they have been in voting theory as a whole, tournament rules have been central to this line of research.} 
This line of work was initiated by \cite{DBLP:conf/aaai/AnshelevichBP15,DBLP:journals/ai/AnshelevichBEPS18}, who showed that the \emph{Copeland} rule has distortion $5$, and established a general lower bound of $3$ for deterministic rules\footnote{Note that 3 is a critical threshold for metric distortion, since if a majority of voters prefers $a$ over $b$ then the cost of $a$ is at most $3$ times the cost of $b$. Distortion 3 can therefore be viewed as a relaxation of a \emph{Condorcet winner} (which may not always exist).}. They conjectured that another rule, \emph{Ranked Pairs} \cite{tideman1987independence}, has distortion $3$ and proved their conjecture in elections with at most $4$ candidates, but it was later disproved by \cite{DBLP:conf/sigecom/GoelKM17,DBLP:conf/aaai/Kempe20a}.

\edit{}{However, this conjecture was later disproved \edit{by}{} \cite{DBLP:conf/sigecom/GoelKM17}, \edit{and}{} \cite{DBLP:conf/aaai/Kempe20a} \edit{showed that Ranked Pairs can have distortion as large as $\sqrt{m/2}$ in $m$-candidate elections.}{}}
The first deterministic improvement over the Copeland rule was by \cite{DBLP:conf/ec/MunagalaW19}, who gave a novel rule called \emph{Weighted Uncovered Set} with distortion $2 + \sqrt{5} \approx 4.236$ \cite{DBLP:conf/ec/MunagalaW19}.
Finally, \cite{DBLP:conf/focs/GkatzelisHS20} proved that there is a deterministic rule with distortion 3, resolving the optimal distortion of deterministic voting rules.
Since this work, several simple deterministic voting rules have been found with optimal distortion $3$ \cite{DBLP:conf/ijcai/KizilkayaK22,DBLP:conf/sigecom/Kizilkaya023}. \edit{}{Still, a deterministic tournament rule with distortion $3$, or any improvement beyond $2 + \sqrt{5}$, remained elusive. }

With deterministic distortion settled, focus shifted to randomized rules. It was conjectured that the optimal distortion is $2$ \cite{DBLP:conf/sigecom/GoelKM17}, but this was disproved independently by \cite{DBLP:conf/soda/CharikarR22} and \cite{pulyassary2021randomized}, with \cite{DBLP:conf/soda/CharikarR22} establishing a lower bound of $2.112$. Recently, \cite{DBLP:journals/jacm/CharikarRWW24} designed a randomized voting rule with distortion at most $2.753$, breaking the longstanding barrier of $3$. Their voting rule uses a mixture of \emph{Maximal Lotteries}, a game-theoretic voting rule which dates back to the 60's \cite{kreweras1965aggregation, fishburn1984probabilistic}, and a randomized variant of \cite{DBLP:conf/ec/MunagalaW19}'s Weighted Uncovered Set. Closing the gap between $2.112$ and $2.753$ remains a major challenge of computational social choice.

Just as they have been in voting theory as a whole, tournament rules have been central to the metric distortion problem---the Copeland rule, Ranked Pairs, Weighted Uncovered Set, and Maximal Lotteries are all tournament rules. And yet, there are persistent gaps in our understanding of the metric distortion of tournament rules. Unlike the general setting, in which deterministic rules are resolved but randomized rules remain open, the situation for tournament rules is reversed: \cite{DBLP:journals/jacm/CharikarRWW24} showed that the Maximal Lotteries rule achieves the optimal randomized distortion of $3$, but for deterministic rules, the gap lies between $3$ \cite{DBLP:conf/sigecom/GoelKM17} and $2 + \sqrt{5}$ \cite{DBLP:conf/ec/MunagalaW19}, as if rolling back the years of progress on the general problem.

Given the historical interplay between metric distortion and tournament rules, closing this gap could have value in both directions. If metric distortion continues to be a reliable compass towards new voting rules, we may uncover natural tournament rules that remain unknown. For example, a deterministic tournament rule with distortion $3$ may be a deterministic analogue of Maximal Lotteries; a natural stand-in for settings where randomness is undesirable. On the other hand, the work of \cite{DBLP:journals/jacm/CharikarRWW24} suggests that such voting rules could aid in closing the gap for general randomized voting rules, just as Maximal Lotteries and Weighted Uncovered Set were key pieces of that work.

%% file: sections/our.tex
\subsection{Our Contributions and Technical Overview}

\paragraph{Metric distortion for deterministic tournament rules.}

Our first contribution improves our understanding of the optimal distortion of deterministic tournament rules. We show that the optimal distortion of tournament rules must lie between $3.1128$ (\Cref{thm:lb}) and $3.9312$ (\Cref{thm:2ub}).

\begin{theorem}\label{thm:tournament-bounds}
The optimal distortion of deterministic tournament rules lies between $3.1128$ and $3.9312$.     
\end{theorem}

We offer a handful of qualitative interpretations of these results. 

First, the lower bound proves a strict separation between the optimal distortion of deterministic and randomized tournament voting rules. Adding to other separations proved for general voting rules \cite{DBLP:journals/ai/AnshelevichBEPS18,DBLP:journals/jacm/CharikarRWW24} and \emph{unweighted} tournament rules\footnote{These are rules, such as the Copeland rule, which only use the binary results of majority votes between pairs of candidates.} \cite{DBLP:journals/ai/AnshelevichBEPS18, DBLP:journals/corr/abs-2403-18340}, these results continue to demonstrate the power of randomization in voting.

Second, while one might expect the distortion of deterministic tournament rules to behave like general deterministic rules, we find that they more closely mirror general randomized rules. Like with randomized rules, our results show that the optimal distortion for deterministic tournament rules is non-integral due to unexpected pathological counterexamples.

Finally, the lower bound eliminates the possibility of a clean deterministic tournament rule with distortion $3$, dashing hopes that such a rule could shed new light on randomized metric distortion. However, our techniques build on the tools developed by \cite{DBLP:journals/jacm/CharikarRWW24} for randomized rules, and allow us to extract more mileage out of them. We are optimistic that these insights may find broader value in metric distortion. Our upper bounds also show that like for general randomized rules, we may also continue to discover interesting new voting rules in narrowing the remaining gap for tournament rules as well.

Our main technical tool in proving these bounds is the \textit{biased metric} framework introduced by \cite{DBLP:conf/soda/CharikarR22} and refined by \cite{DBLP:journals/jacm/CharikarRWW24}. This framework characterizes the worst-case metric spaces for the problem, and provides necessary and sufficient conditions for low metric distortion (\Cref{thm:biased-iff}). In particular, this condition is an inequality between two integrals of functions built by considering what fraction of voters satisfy certain preferences. 

We give a geometric interpretation of these functions, which we view as a stack of blocks. This physical analogy is particularly helpful for reasoning about these conditions, since we can ``move'' blocks around to ensure that they cover a certain area. We use this approach to derive sufficient conditions for the existence of a low distortion candidate that can be verified based on local properties of the tournament graph (\Cref{lem:post-shift}). This condition in turn leads to a notion we call \emph{blanketing}, a stronger version of a covering condition in the social choice literature used to derive the \emph{Uncovered Set} \cite{fishburn1977condorcet,miller1980new,moulin1986choosing} and \emph{Weighted Uncovered Set} \cite{DBLP:conf/ec/MunagalaW19} voting rules. Our new voting rule, \emph{\nameref{box:unblanketed}} (\cpageref{box:unblanketed}), selects a candidate $j^*$ that is not blanketed (for appropriate choice of parameters) by any other candidate. Using a graph theoretic argument, we show that such a candidate must always exist. 

One challenge in proving a lower bound is a result due to \cite{DBLP:conf/ec/MunagalaW19} that if an election has a cyclically symmetric tournament graph, then every candidate has distortion 3. Therefore, any improved lower bound must use an inherently asymmetric election, in contrast to the hardest instances for deterministic rules \cite{DBLP:journals/ai/AnshelevichBEPS18}, and randomized tournament rules \cite{DBLP:conf/sigecom/GoelKM17}. 

A natural starting point is to understand what kinds of instances do not satisfy the sufficient conditions used for our upper bound (\Cref{lem:post-shift}). We find that these conditions can fail when candidates lie along a path with exponentially increasing weights, since these conditions bound the distortion using ratios of these weights. Following this idea, we construct a cycle on 5 candidates where the weights follow an increasing geometric sequence (\Cref{fig:5lb}), and the goal is to show that each candidate can have high distortion when the candidate preceding it is optimal in a certain biased metric. The increasing weights make it easy to give a lower bound for all but one of the candidates (the one whose predecessor beats it by the least). For most of the candidates, the underlying metric that we use comes from the subclass of biased metrics called \textit{$(0,1,2,3)$-metrics} used by \cite{DBLP:conf/soda/CharikarR22} to give improved lower bounds against randomized voting rules. However, for the last candidate, we use a ``half-integral'' version of these metrics. This contradicts their conjecture that the $(0,1,2,3)$-metrics are the hardest metrics, and opens up the possibility of improving their lower bounds as well. \\

The limitations of tournament rules in the metric distortion setting naturally prompt the question: what can we do by incorporating information just beyond pairwise comparisons? What additional information is useful enough for better voting rules to be possible?

With this in mind, we introduce \emph{$k$-tournament rules} as a generalization of tournament rules. These rules elicit information based on $k$-wise comparisons: for each tuple of $k$ candidates $(c_1, c_2, \ldots , \allowbreak c_k)$, they ask how many voters have the preference $c_1 \succ c_2 \succ \cdots \succ c_k$. $k$-tournaments capture preference elicitation methods in some RLHF settings: for example, in the work of \cite{DBLP:conf/nips/Ouyang0JAWMZASR22}, they ``present labelers with anywhere between $k = 4$ and $k = 9$ responses to rank.'' Similarly, the work of \cite{zhu2023principled} develops a theory of RLHF under $k$-wise preferences.

\edit{}{Though this generalization feels natural and has seen practical use, there is little published discussion on these rules, possibly due to the historical absence of a framework that can guide the design and evaluation of $k$-tournaments.} %

\paragraph{$k$-tournament rules and their metric distortion.} In our second contribution, we develop the first $k$-tournament voting rules\edit{}{, generalizing tournament rules,} and show that the additional information afforded to these rules allows them to have improved distortion, both in the deterministic and randomized settings.

\begin{theorem}\label{thm:k-tournament-bounds}
The optimal distortion of deterministic $k$-tournament rules is $3 + \tilde{O}\big(\frac{1}{k^{1/4}}\big)$. The optimal distortion for randomized $k$-tournament rules is $3 - \Omega(1)$, even for $k = 3$.     
\end{theorem}

An intriguing feature of these bounds is that they actually require even less information than $k$-tournament rules would provide. Rather than the full rankings over tuples of $k$ candidates, they only need to know for each set of $k$ candidates, how often each candidate is a voter's favorite or least favorite within the set. We view this as a potentially practical takeaway from our result: in settings (like \cite{DBLP:conf/nips/Ouyang0JAWMZASR22,zhu2023principled}) where one might hope to achieve better performance by querying voters over sets of $k$ candidates instead of just pairs, asking for their top and bottom choices in the set might be sufficient for high performance, without demanding as high a cognitive load as full rankings.

The key insight in our voting rules lies in a connection to a different problem: \emph{committee selection}. In this problem, rather than choosing a single winner, we would like to choose a committee (or distribution over committees) of $k$ winners such that no outside candidate is preferred over the committee by a large fraction of voters. Then a $k$-committee selection algorithm can be viewed as a $(k + 1)$-tournament rule, but with $k$ winners. We can use this correspondence to leverage techniques in committee selection to develop $k$-tournament voting rules.

We apply stable lotteries \cite{DBLP:journals/teco/ChengJMW20}, a randomized committee selection algorithm originally designed for proportional fairness guarantees, to our study of $k$-tournaments.
A stable $k$-lottery is a distribution on candidates obtained from the equilibrium of a 2-player game, one who picks a committee of $k$ candidates and another, who picks a single candidate, each hoping that voters will prefer their proposed candidate(s) over the other player's proposal.
We first prove a technical lemma on the distribution of stable $k$-lotteries: \nameref{lem:main}. Intuitively, if for a subset $S$ of candidates, many voters think that a single candidate $i \notin S$ is better than all of $S$, then the probability distribution of a stable $k$-lottery cannot put too much weight on candidates in $S$. 
Such a characterization fits well with the biased metric framework, enabling us to analyze voting rules that use stable lotteries as components.

We next propose two $k$-tournament rules: a deterministic rule of \emph{\nameref{box:simul_veto}} (\cpageref{box:simul_veto}) and a randomized rule of \emph{\nameref{box:two_lotteries}} (\cpageref{box:two_lotteries}).

For the deterministic $k$-tournament rule of \nameref{box:simul_veto}, we show that as we increase $k$, its distortion can get arbitrarily close to $3$ (\cref{thm:distortion_simul_veto}), which is the performance of the optimal deterministic ranked voting rule. 
\nameref{box:simul_veto} starts with a pruning procedure on the set of candidates (called \nameref{box:quasi-kernel}, first introduced and used by \cite{DBLP:journals/jacm/CharikarRWW24}) to come up with a shortlist $\hat{C}$. Next we run an elimination process to determine the winner, a $k$-tournament version of Simultaneous Veto \cite{DBLP:conf/sigecom/Kizilkaya023}.
Each candidate in $\hat{C}$ starts out with an initial score which is reduced continuously over time and the last candidate remaining with positive score is declared the winner.
The procedure uses the $k$-tournament information given to it to compute starting scores for all candidates in $\hat{C}$.
In particular, the initial scores are exactly the weights in the probability distribution of a stable $(k - 1)$-lottery on $\hat{C}$. The rate at which the scores are decreased changes dynamically, and is given by the weights of a reverse stable $(k - 1)$-lottery on the set of candidates with non-zero scores. This reverse stable $(k - 1)$-lottery is a stable $(k - 1)$-lottery computed by reversing the voter rankings of candidates.
In analyzing the metric distortion of this rule using the biased metric framework, we need good bounds on the initial scores of candidates and the rates at which the scores decrease. For the initial scores, we need to relate the total probability mass placed by the stable $(k - 1)$-lottery on candidates in a subset $J$ to the fraction of voters who prefer a candidate $i \notin J$ to all of $J$.
For the rates of score decrease, we need to relate the probability mass placed by the reverse stable $(k - 1)$-lottery on candidates in a subset $I$ to the fraction of voters who prefer all of $I$ to a candidate $j \notin I$.
Both these bounds come from \nameref{lem:main}.

For the randomized $k$-tournament rule of \nameref{box:two_lotteries}, we show that for any $k \geq 3$, its distortion is strictly less than $3$ (\cref{thm:distortion_two_lotteries}), which even among ranked voting rules is only surpassed by the recent work of \cite{DBLP:journals/jacm/CharikarRWW24}. In particular, combined with the fact that the distortion of any randomized ($2$-)tournament rules must be at least $3$ \cite{DBLP:conf/sigecom/GoelKM17}, our result shows that a distortion separation between the classes of randomized tournament rules and randomized $3$-tournament rules. The design of \nameref{box:two_lotteries} closely follows the work of \cite{DBLP:journals/jacm/CharikarRWW24}, and the difference is that we use a draw from a stable lottery to replace the favorite candidate of a random voter.

%% file: sections/related.tex
\subsection{Related Work}
\label{sec:related}

\subsubsection{Tournaments and Social Choice}
Tournaments have been a prominent topic in social choice theory. Surveys including \cite{brandt2016tournament,fischer2016weighted,DBLP:conf/ijcai/Suksompong21} provide a comprehensive overview, and here we only discuss some of the closest-in-topic lines of research.

Tournaments can come in two forms---unweighted and weighted---and can be represented by tournament graphs (whose vertices correspond to the candidates) with unweighted or weighted edges. The unweighted version captures majority relations and contains the information of whether a majority of voters prefer $a$ to $b$ for each pair of candidates $a$ and $b$. The weighted version includes additional information on the margins of victory. Fishburn \cite{fishburn1977condorcet} classifies voting rules based on the information they use. He calls a voting rule that relies only on the unweighted tournament a \emph{C1 rule}, and a rule that relies on the weighted tournament (but nothing more) a \emph{C2 rule}. In this language, the novel voting rules in this work, including \nameref{box:unblanketed}, \nameref{box:two_lotteries}, and \nameref{box:simul_veto}, are all C2 rules.

\paragraph{Tournaments from rankings.}
There are $2^{\binom{m}{2}}$ possible unweighted tournament graphs in an election with $m$ candidates (assuming no ties). A line of research has investigated whether all these unweighted tournament graphs can be realized by the pairwise majority relations of $n$ voters with ordinal rankings of the candidates. The answer is positive---any such graph can indeed be realized \cite{mcgarvey1953theorem}. However, quantitatively, realizing such a graph may require $n = \Theta\big(\frac{m}{\log m}\big)$ voters, and this bound is known to be tight \cite{stearns1959voting,erdos1964representation}. Moreover, Alon \cite{alon2002voting} shows that any \mbox{unweighted} tournament graph can be realized with each margin of victory being at least $\frac{1}{2} + \Theta\big(\frac{1}{\sqrt{m}}\big)$ versus $\frac{1}{2} - \Theta\big(\frac{1}{\sqrt{m}}\big)$, and this bound is tight.

\paragraph{Voting rules  with $k$-queries.} The very recent work of \cite{Halpern2024computing} is the only prior work we are aware of that can be viewed as studying $k$-tournament voting rules. Their work characterizes the number of queries needed to implement certain voting rules, where each query returns the preference distribution among voters on a size-$k$ subset of candidates. They show that some scoring rules such as Plurality, and other voting rules such as Single Transferable Vote (STV), are not implementable with any number of queries of bounded size---in our terms, these rules are \textit{not} $k$-tournament rules for any finite $k$. These limitations motivate designing specialized voting rules that effectively leverage $k$-tournament information. In our work, we introduce such a family of voting rules, and show that they can achieve low metric distortion.

\paragraph{$k$-tournaments.} 
Mathematicians have proposed the concept of \emph{$k$-tournaments} \cite{marshall1994properties,marshall2002regular,petrovic2006edge} (also called \emph{$k$-hypertournaments} \cite{gutin1997hamiltonian,guofei2000score,DBLP:journals/jgt/AiGG22}) as a generalization of \mbox{unweighted} tournament graphs. In these hypergraphs, for each size-$k$ subset $S$ of vertices, exactly one of the $k!$ possible orientations of the hyperedge on $S$ is included in the hypergraph. Our definition of $k$-tournament rules is aligned with this existing notion but allows for fractional hyperedges: instead of requiring exactly one orientation to be included, we require the total (nonnegative) weight of these orientations to be one.

\paragraph{Tournaments and reinforcement learning from human feedback.} 
Reinforcement learning from human feedback (RLHF) \cite{christiano2017deep} is a widely used method for finetuning large language models (e.g., \cite{achiam2023gpt,touvron2023llama}).
Here, a pre-trained foundation model is aligned with human values or trained for a specific task using human feedback on model outputs.
Recently, there has been a surge of interest in applying methods from social choice theory to analyzing RLHF \cite{ge2024axioms,conitzer2024social,dai2024mapping,mishra2023ai,zhong2024provable,park2024principled,chakraborty2024maxmin,swamy2024minimaximalist,siththaranjan2024distributional}.
While most applications of RLHF use human feedback in the form of comparisons between pairs of model outputs (i.e., information captured by a tournament), some applications use rankings for small sets of model responses (i.e., $k$-tournament information) \cite{DBLP:conf/nips/Ouyang0JAWMZASR22}.
Recent work \cite{zhu2023principled} analyzed statistical properties of algorithms that use such pairwise and $k$-wise comparison information. 

\subsubsection{Distortion in Social Choice}

Procaccia and Rosenschein \cite{DBLP:conf/cia/ProcacciaR06} were the first to introduce the notion of \emph{distortion} in social choice, with the aim of quantifying the welfare loss associated with various voting rules. Distortion can also be used to quantify the welfare loss from the information limitations inherent in classes of voting rules \cite{DBLP:journals/ai/AnshelevichBEPS18,boutilier2016incomplete}. We refer the reader to the survey of \cite{DBLP:conf/ijcai/AnshelevichF0V21} for a more comprehensive overview of the field.

The metric assumption was first introduced by \cite{DBLP:journals/ai/AnshelevichBEPS18} and enables many social choice rules to achieve constant metric distortion. Most research on metric distortion is focused on ranked voting rules. Among deterministic rules, the optimal metric distortion is $3$ \cite{DBLP:conf/focs/GkatzelisHS20,DBLP:conf/ijcai/KizilkayaK22,DBLP:conf/sigecom/Kizilkaya023}. Among randomized ones, the simple rule of random dictatorship achieves metric distortion $3$ as well \cite{DBLP:journals/jair/AnshelevichP17,DBLP:conf/sigecom/FeldmanFG16}, and the optimal metric distortion is known to be between $2.112$ \cite{DBLP:conf/soda/CharikarR22} and $2.753$ \cite{DBLP:journals/jacm/CharikarRWW24}.

The classes of unweighted and weighted tournament rules are particularly relevant to our work. %
Among deterministic unweighted tournament rules, the classical Copeland rule achieves the optimal metric distortion of $5$ \cite{DBLP:journals/ai/AnshelevichBEPS18}. Among randomized unweighted tournament rules, the C1 Maximal Lotteries rule was recently shown to be optimal, achieving metric distortion of $4$ \cite{DBLP:journals/corr/abs-2403-18340}. As discussed earlier, within the class of weighted tournament rules, the C2 Maximal Lotteries rule is optimal among randomized ones with metric distortion $3$, and our work improves the distortion bound for the optimal deterministic rule from $[3, 2 + \sqrt{5}]$ to $(3.112, 3.932)$.

Several works have quantitatively studied the tradeoff between distortion and the amount of information available to voting rules, such as \cite{DBLP:conf/aaai/Kempe20b,DBLP:journals/jair/AnagnostidesFP22,DBLP:conf/ijcai/Ebadian0M24} within the metric distortion framework and \cite{DBLP:conf/nips/MandalPSW19,DBLP:conf/sigecom/MandalSW20} without the metric assumption. \cite{goel2025metric} use metric distortion to analyze models where voters can deliberate in groups of size $k$, and show that the additional information resulting from deliberations can significantly improve distortion guarantees, even if the groups are quite small.

\subsubsection{Stable Lotteries}

Stable lotteries \cite{DBLP:journals/teco/ChengJMW20} play a crucial role in our design and analysis of $k$-tournament rules. The notion of stable lotteries stems from the line of research on proportional fairness in committee selection and participatory budgeting, and serves as a generalization of core-stable committees. The aim in these settings is to select a subset of candidates to form a committee that can fairly represent all the voters, subject to a budget constraint on the number (or total weight) of selected candidates. A core-stable solution \cite{foley1970lindahl,aziz2017justified,DBLP:conf/sigecom/FainMS18} is one in which for every parameter $\alpha \in (0, 1]$, no $\alpha$ fraction of voters $S$ can find an alternative solution using $\alpha$ fraction of the total budget so that all these voters $S$ strictly prefer the new solution over the old one. Core stability provides a strong guarantee of proportional fairness and implies notions including justified representation (JR) and extended justified representation (EJR) in approval voting \cite{aziz2017justified,DBLP:conf/nips/PierczynskiSP21}. Unfortunately, in many settings, core-stable committees do not exist (or sometimes are not known to exist) \cite{DBLP:conf/sigecom/FainMS18,DBLP:journals/teco/ChengJMW20}, and the notion of stable lotteries is a relaxation that always exists in several of such settings \cite{DBLP:journals/teco/ChengJMW20}. Stable lotteries have been applied as a technical tool in other social choice research such as in showing the existence of approximate core-stable committees \cite{jiang2020approximately} and in deriving the optimal (non-metric) distortion of randomized ranked voting rules \cite{DBLP:journals/teco/EbadianKPS24}. In contrast to our work that establishes new properties of stable-lottery solutions (\nameref{lem:main}), these works apply the existence and proportional fairness guarantees of stable lotteries as a black box.

%% file: sections/preliminaries.tex
\section{Preliminaries}
\label{sec:prelim}

\subsection{Social Choice and \texorpdfstring{$k$}{k}-Tournaments}
\label{subsec:prelim/tournament}

We first describe a general classical social choice setting. Let there be a set $V$ of $n$ voters and a set $C$ of $m$ candidates. We often use $v$ to denote a voter, and numbers, indices $i, j, k$, or $a, b$ to denote candidates.\footnote{Note that $k$ is overloaded, used for $k$-tournaments in \Cref{sec:sr-lemma,sec:k_tournament} in line with the committee selection literature, and denoting a candidate in \Cref{sec:biased,sec:det2} in line with \cite{DBLP:journals/jacm/CharikarRWW24}. These uses are strictly separated to avoid confusion.} We specifically use $i^*$ to denote the intended optimal candidate, and $j^*$ to denote the candidate a voting rule intends to choose.

Each voter has their preference over the candidates in the form of a total order. The preferences of all voters form a \emph{preference profile}. The notation $a \succ_v b$ denotes that voter $v$ prefers candidate $a$ to candidate $b$. We may drop $v$ in the shorthand $a \succ b$ when it is clear from context. A \emph{social choice rule}, or \emph{voting rule}, selects a candidate in the set $C$ as the winner. Formally, a deterministic social choice rule maps information on voter preferences to a winning candidate, and a randomized social choice rule maps that information to a distribution over winning candidates.

We use the notation $S_{P}$ to denote the subset of voters that satisfy condition $P$, %
and use $s_{P}$ to denote the fraction of voters that satisfy condition $P$. In other words, $S_{P} = \{v \in V : v \text{ satisfies } P\}$ and $s_{P} = \frac{1}{n}|S_{P}|$. For example, $s_{a \succ b}$ denotes the fraction of voters who prefer $a$ to $b$.

In our specific setting, the information on voter preferences that a social choice rule has access to is in the form of \emph{$k$-tournaments}. 
\begin{definition}[$k$-tournament rules]
A deterministic (respectively, randomized) \emph{$k$-tournament rule} is a function that maps the information
\[
\left(s_{c_1 \succ c_2 \succ \cdots \succ c_k}\right)_{c_1, c_2, \ldots, c_k \in C}
\]
to a winning candidate (respectively, a distribution over winning candidates).
\end{definition}

When discussing $k$-tournament rules, we assume the number of candidates $m$ is at least $k$. (Alternatively, we can redefine $k$-tournament rules to be $m$-tournament rules for $k > m$.) It is easy to see that $k$-tournament rules are also $(k + 1)$-tournament rules. We refer to $2$-tournament rules as \emph{tournament rules}.

\subsection{Metric Distortion}
\label{subsec:prelim/distortion}

In the metric distortion framework, there is a (pseudo-)metric space $(X, d)$ where the points are $X = V \cup C$, and the distance function $d \colon X \times X \to \mathbb{R}_{\geq 0}$ defines the distance between each pair of points in the metric space. It satisfies the following conditions.
\begin{itemize}
\item $d(x, x) = 0$ for all $x \in X$.
\item $d(x, y) = d(y, x)$ for all $x, y \in X$ (symmetry).
\item $d(x, y) + d(y, z) \geq d(x, z)$ for all $x, y, z \in X$ (triangle inequality).
\end{itemize}

Each voter is assumed to prefer closer candidates to farther ones: $a \succ_v b$ if $d(v, a) < d(v, b)$, and break ties in any consistent way if $d(v, a) = d(v, b)$.

The \emph{social cost} of any candidate is the sum of its distances to all voters: $\SC(c) = \sum_{v \in V} d(v, c)$.
\begin{definition}[metric distortion]
The \emph{metric distortion} of a voting rule $f$ is the minimum number $\alpha \in [1, +\infty]$ so that for all metric spaces, the (expected, if $f$ is randomized) social cost of the selected candidate of $f$ is at most $\alpha$ times the social cost of any candidate.
\end{definition}

%% file: sections/biased.tex
\section{The Biased Metric Framework}
\label{sec:biased}

The approach used by \cite{DBLP:conf/soda/CharikarR22,DBLP:journals/jacm/CharikarRWW24} was to determine the set of worst-case metrics for any voting rule, given a fixed election instance. They showed that each such metric can be parametrized by a vector $(x_1, \dots, x_m) \in \mathbb{R}_{\geq0}^m$. Semantically $x_j$ represents the distance between $j$ and the optimal candidate $i^*$, and the following metrics minimize the distances between the voters and the optimal candidate, while maximizing the distances from the voters to the other candidates.

\begin{definition}\label{def:biased}
Let $(x_1, \ldots, x_m)  \in \mathbb{R}_{\geq0}^m$  such that $x_{i^*} = 0$ for some $i^*$. Given an election instance, the \emph{biased metric} for the vector $(x_1, \ldots, x_m)$ is defined as follows. For a voter $v$ and candidates $i^*$ and $j$, let
\begin{align*}
d(i^*, v) &= \frac12 \max_{i, j: j \cgeq_v i} (x_j - x_i),\\
d(j, v) - d(i^*, v) &= \min_{k: j \cgeq_v k} x_k.
\end{align*}
\end{definition}

\cite{DBLP:conf/soda/CharikarR22} showed that given any election instance and any voting rule (i.e., a distribution over candidates), some biased metric maximizes the distortion of the voting rule. 

For a fixed biased metric parametrized by $(x_1, \ldots, x_m)$, \cite{DBLP:journals/jacm/CharikarRWW24} showed that the expected difference in social cost between the candidate output by a voting rule giving probability $p_j$ to candidate $j$ and the optimal candidate $i^*$ can be written as 
$$\sum_{j \in C} (\SC(j) - \SC(i^*))p_j  = \int_0^\infty \sum_{j\notin I_t}  s_{I_t \cg j} p_j \dif t$$ 
where $I_t := \{k \in C: x_k \leq t\}$. (Note that similarly we define $J_t := I_t^c = \{k \in C: x_k > t\}$.) On the other hand, we can express the social cost of the optimal candidate as
$$2\SC(i^*) =  \int_0^\infty(1 - s_{\forall j\cg i,  x_j - x_i \leq t}) \dif t.$$

We refer the reader to \cite{DBLP:journals/jacm/CharikarRWW24} for a full explanation, but we offer an intuitive visual way of viewing these integrals (see \Cref{fig:stacked-blocks}), which also suggests ways of working with them to get bounds on the distortion.  %
\begin{figure}[ht!] %
  \centering
  \includegraphics[width=0.9\textwidth]{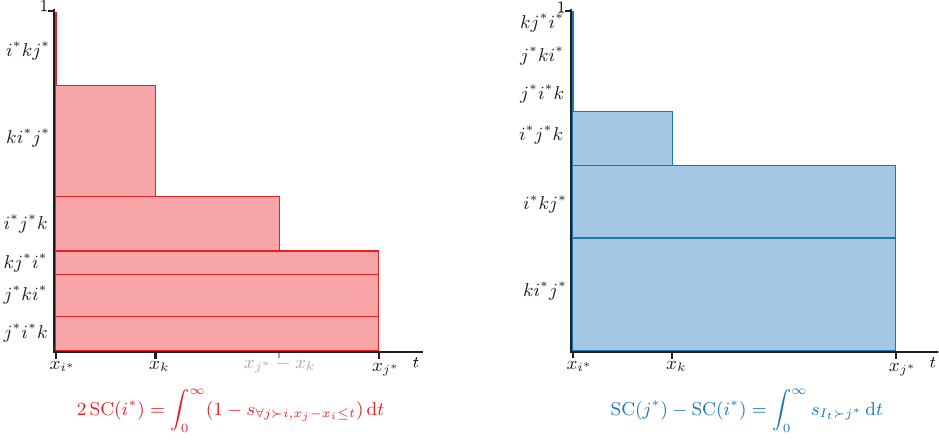} 
  \caption{The biased metric integrals as stacked blocks. Each block is annotated with its corresponding ranking on the left, and voters with the same ranking are aggregated into a single block. The figures are exact representations of the integrals when there are only three candidates $\{i^*, j^*, k\}$, but when there are more candidates, the figure on the left can be viewed as a lower bound, and the figure on the right can be viewed as an upper bound. For example, if $i^*\cg_v j^*\cg_v k$, then $2d(i^*, v)\geq x_{j^*} - x_k$ and $d(j^*,v) - d(i^*, v) \leq x_k$.} 
  \label{fig:stacked-blocks}
\end{figure}

In order to compare $2\SC(i^*)$ and $\SC(j^*) - \SC(i^*)$, we fix a biased metric for which $x_{i^*} = 0$. We represent $2\SC(i^*)$ in the following way. Each voter $v$ corresponds to a block with length $2d(i^*, v) =  \max_{i, j: j \cgeq_v i} (x_j - x_i)$ and height $\frac1n$. $2\SC(i^*)$ is precisely the total area of these blocks. We then stack these blocks on top of each other, in decreasing order of height, which allows us to represent $2\SC(i^*)$ as the integral of a decreasing function. In particular, the function maps $t$ to the fraction of voters (total height of blocks) whose blocks are length greater than $t$. This function is precisely $1 - s_{\forall j\cg i,  x_j - x_i \leq t}$.

Similarly, we can represent $\SC(j^*) - \SC(i^*)$ as the stack of blocks where each voter $v$ corresponds to a block with length $d(j^*, v) - d(i^*, v) = \min_{k: j \cgeq_v k} x_k$, and height $\frac1n$. $\SC(j^*) - \SC(i^*)$ can be represented as the integral of a decreasing function, which maps $t$ the fraction of voters (total height of blocks) whose blocks are length greater than $t$, which is precisely $s_{I_t \cg j^*}$. Note that when considering a distribution over candidates which chooses $j$ with probability $p_j$, we take the weighted average of these functions, which maps $t$ to $\sum_{j\in C} s_{I_t \cg j}p_j$ (note $s_{I_t \cg j} = 0$ for $j \in I_t$).

\vspace{11pt}

Though it is a simple idea, this interpretation of the costs in terms of stacked blocks is a powerful way of understanding the biased metrics and how we can derive distortion bounds from them. 

For example, one tool that was central in finding deterministic voting rules with distortion 3 was the notion of the graph $G(j^*, i^*)$ for a pair of candidates $j^*, i^*$. Defined by \cite{DBLP:conf/ec/MunagalaW19} (and later specialized to the notion of the \textit{domination graph} by \cite{DBLP:conf/focs/GkatzelisHS20}), $G(j^*, i^*)$  is a bipartite graph, each side of which is a copy of the set of voters $V$, and there is an edge from from $v$ to $v'$ if there exists a candidate $k$ such that $j^* \cgeq_v k$ and $k \cgeq_{v'} i^*$. \cite{DBLP:conf/ec/MunagalaW19} showed that if $G(j^*, i^*)$ has a perfect matching, then $\SC(j^*) \leq 3\SC(i^*)$  (\cite{DBLP:conf/focs/GkatzelisHS20} showed that there always exists a candidate $j^*$ such that $G(j^*, i^*)$ has a perfect matching for each $i^* \neq j^*$). 

\cite{DBLP:journals/jacm/CharikarRWW24} showed that the biased metrics can be used to recover the \cite{DBLP:conf/ec/MunagalaW19}'s result, and this approach can be viewed elegantly through the perspective of the stacked blocks. In particular, a matching in $G(j^*, i^*)$ corresponds exactly with a matching of blocks in $\SC(j^*) - \SC(i^*)$ with longer blocks in $2\SC(i^*)$. That is, if $j^* \cgeq_v k$ and $k \cgeq_{v'} i^*$ (so that $(v, v')$ is an edge in $G(j^*, i^*)$)  then
$$d(j^*, v) - d(i^*, v) \leq x_k = x_k - x_{i^*} \leq 2d(i^*, v')$$
meaning that the block corresponding to $v$ in $\SC(j^*) - \SC(i^*)$ is shorter than the block corresponding to $v'$ in $2\SC(i^*)$.

All together, the biased metric framework can be neatly summarized in the following theorem, which bounds the distortion of a (possibly randomized) voting rule on a fixed election instance.

\begin{theorem}[necessary and sufficient integral condition \cite{DBLP:journals/jacm/CharikarRWW24}]\label{thm:biased-iff}
Let $\Delta$ be the distribution over candidates that chooses candidate $j$ with probability $p_j$. $\E_{j^*\sim \Delta}[\SC(j^*)] \leq (1 + 2\lambda)\SC(i^*)$ if and only if for all  $(x_1, \dots, x_m) \in \mathbb{R}_{\geq0}^m$ such that $x_{i^*} = 0$,
\begin{equation}\label{eq:precise-biased}
\int_0^\infty \sum_{j\notin I_t}  s_{I_t \cg j} p_j \dif t \leq  \lambda\int_0^\infty(1 - s_{\forall j\cg i,  x_j - x_i \leq t}) \dif t  
\end{equation}
where $I_t := \{k \in C: x_k \leq t\}$. In particular, $\SC(j^*) \leq (1 + 2\lambda)\SC(i^*)$ if and only if for all  $(x_1, \dots, x_m) \in \mathbb{R}_{\geq0}^m$ such that $x_{i^*} = 0$,
\begin{equation}\label{eq:precise-biased-det}
\int_0^\infty  s_{I_t \cg j^*}  \dif t \leq  \lambda\int_0^\infty(1 - s_{\forall j\cg i,  x_j - x_i \leq t}) \dif t.
\end{equation}
\end{theorem}

We can use \Cref{thm:biased-iff} to derive sufficient conditions for a candidate (and thereby a voting rule) to have small distortion. One straightforward and particularly useful condition is below. These conditions will be useful for \cref{sec:sr-lemma,sec:k_tournament}.

\begin{corollary}[\cite{DBLP:conf/soda/CharikarR22}]\label{cor:lp-condition}
 $\SC(j^*) \leq (1 + 2\lambda)\SC(i^*)$ if for all partitions of the candidates $I \sqcup J = C$ such that $i^* \in I$ and $j^* \in J$, 
\begin{equation}\label{eq:weaker-biased-det}
s_{I \cg j^*}  \leq  \lambda(1 - s_{i^* \cg J}).
\end{equation}
Similarly, $\sum_{j\in C} p_j\SC(j) \leq (1 + 2\lambda)\SC(i^*)$ if for all partitions of the candidates $I \sqcup J = C$ such that $i^* \in I$, 
\begin{equation}\label{eq:weaker-biased-rand}
\sum_{j \in J}p_j s_{I \cg j}  \leq  \lambda(1 - s_{i^* \cg J}).
\end{equation}
\end{corollary}

\begin{proof}[Proof sketch]   
To prove the corollary, it suffices to show that $1 - s_{i^* \cg J_t} \leq 1 - s_{\forall j\cg i,  x_j - x_i \leq t}$. If for a voter $v$, whenever $j\cg i$ we have $x_j - x_i \leq t$, then we must have $i^*\cg j$ for all $j$ such that $x_j > t$. This means exactly that $ S_{\forall j\cg i,  x_j - x_i \leq t} \subseteq S_{i^* \cg J_t}$. 
\end{proof}

In the view of \cite{DBLP:conf/soda/CharikarR22}, the conditions of \Cref{cor:lp-condition} define a linear program which can be used to determine the distortion of a given voting rule.

%% file: sections/determ_two.tex
\NewEnviron{hideproofAone}[1][Proof]{%
  \ifshowproofs
    \begin{proof}[#1]
      \BODY
    \end{proof}
  \else
    \par\noindent\textit{#1.} See \Cref{prf:tourn-biased}.\hfill$\square$\par\nobreak\vskip\medskipamount
  \fi
}

\NewEnviron{hideenremark}[1][Proof]{%
  \par\noindent\textit{#1.} See \Cref{prf:tourn-biased}. \hfill$\square$%
  \par\nobreak\vskip\medskipamount
}

\section{Deterministic Selection in (\texorpdfstring{$2$}{2}-)Tournaments}
\label{sec:det2}

\subsection{Deriving Tournament Conditions with Biased Metrics}\label{sec:tourn-biased}

A natural way of deriving low distortion tournament rules using \Cref{thm:biased-iff} is to find sufficient conditions implying \cref{eq:precise-biased-det} which can be verified solely using information in the tournament graph. If we can show that candidates satisfying these conditions always exist, then we have a voting rule with distortion at most $1 + 2\lambda.$

As an illustrative example, the following sufficient condition can be derived straightforwardly from \Cref{cor:lp-condition}.

\begin{lemma}\label{lem:lp-tournament}
$\SC(j^*) \leq (1 + 2\lambda) \SC(i^*)$ if for all partitions of the candidates $I \sqcup J = C$ such that $i^* \in I$ and $j^* \in J$, 
$$\min_{i\in I} s_{i\cg j^*} \leq \lambda \max_{j\in J} s_{j\cg i^*}.$$
\end{lemma}

\begin{hideproofAone}[Proof of \cref{lem:lp-tournament}]
If the condition of the lemma is true, then for all partitions of the candidates $I \sqcup J = C$ such that $i^* \in I$ and $j^* \in J$, we have
$$s_{I\cg j^*} \leq \min_{i\in I} s_{i\cg j^*} \leq \lambda \max_{j\in J} s_{j\cg i^*} \leq \lambda(1 - s_{i^*\cg J}).$$
Therefore, (\ref{eq:weaker-biased-det}) is satisfied, and the result follows from \cref{cor:lp-condition}. 
\end{hideproofAone}

This lemma alone is not especially powerful, but we note a couple of interesting consequences. First, if $s_{i^* \cg j^*} \leq \frac{\lambda}{1 + \lambda}$ then $\SC(j^*) \leq (1 + 2\lambda) \SC(i^*)$ (this result is implicit in \cite{DBLP:conf/ec/MunagalaW19}, and is a corollary of \cite[Corollary~5.1]{DBLP:conf/aaai/Kempe20a}). Second, this lemma can be used to recover a result of \cite{DBLP:conf/ec/MunagalaW19}, that in an election with a \textit{cyclically symmetric} tournament graph, every candidate has distortion at most 3. See \Cref{thm:cyc-sym} for the details.

Using the stacked blocks interpretation of the biased metrics, we can prove the following stronger version of \Cref{cor:lp-condition}, which in turn allows us to improve the distortion for deterministic tournament rules.

\begin{lemma}\label{lem:post-shift}
$\SC(j^*) \leq (1 + 2\lambda) \SC(i^*)$ if there exists a candidate $k$ such that either $k = i^*$ or
$$ s_{i^* \cg j^*} \leq \lambda s_{k\cg i^*},$$
and for all partitions of the candidates $I \sqcup J = C$ such that $i^*, k \in I$ and $j^* \in J$, 
$$s_{I\cg j^*} \leq \lambda (1 - s_{i^*, k\cg J}).$$
\end{lemma}

Note the following immediate corollary only using edges in the tournament graph, by the facts that $s_{I\cg j^*} \leq \min_{i\in I} s_{i\cg j^*}$ and $1 - s_{i^*, k\cg J} \geq \max_{j\in J} \{s_{j\cg i^*},  s_{j\cg k}\}.$ This corollary can be used to show a result of \cite{DBLP:journals/ai/AnshelevichBEPS18} that \emph{Ranked Pairs} has distortion at most 3 when there are at most 4 candidates (see \Cref{thm:rp}).

\begin{corollary}\label{cor:post-shift}
$\SC(j^*) \leq (1 + 2\lambda) \SC(i^*)$ if there exists a candidate $k$ such that either $k = i^*$ or
$$ s_{i^* \cg j^*} \leq \lambda s_{k\cg i^*},$$
and for all partitions of the candidates $I \sqcup J = C$ such that $i^*, k \in I$ and $j^* \in J$, 
$$\min_{i\in I} s_{i\cg j^*} \leq \lambda \max_{j\in J} \{s_{j\cg i^*},  s_{j\cg k}\}.$$
\end{corollary}

We present here the following corollary, which we will need in \cref{sec:k_tournament}.
\begin{corollary}\label{cor:two-step}
$\SC(j^*) \leq (\frac{4}{\theta} - 3) \SC(i^*)$ if there exists a candidate $k$ such that $s_{j^* \cg k} \geq \theta$ and either $k = i^*$ or $s_{k \cg i^*} \geq \theta$.
\end{corollary}

\begin{hideproofAone}[Proof of \Cref{cor:two-step}]
We apply \cref{cor:post-shift} with $\lambda = \frac{2}{\theta} - 2$. First, $s_{i^* \cg j^*} \leq \lambda s_{k\cg i^*}$ holds (when $k \neq i^*$) since
\[
s_{i^* \cg j^*} \leq 2 - s_{j^* \cg k} - s_{k \cg i^*} \leq 2 - 2\theta = \lambda \theta \leq \lambda s_{k \cg i^*}.
\]
For the other condition,
\[
\min_{i\in I} s_{i\cg j^*} \leq s_{i^* \cg j^*} \leq \lambda \theta \leq \lambda s_{j^* \cg k} \leq \lambda \max_{j\in J} \{s_{j\cg i^*},  s_{j\cg k}\}. \qedhere
\]
\end{hideproofAone}

\begin{hideproofAone}[Proof of \Cref{lem:post-shift}]
We note that when $k = i^*$, the lemma degenerates to \Cref{cor:lp-condition}. Henceforth, assume that $k \neq i^*$.

We will start by giving a detailed outline of the proof, using the stacked blocks interpretation of biased metrics. Considering the integral on the right side of \cref{eq:precise-biased}, we know that the integrand is above $s_{k\cg i^*}$ up to $t \leq x_k$. (If $k \cg_v i^*$ then the block corresponding to $v$ has length at least $x_k$.) This bound may be the best we can give in terms of edges in the tournament graph, but there may be much excess area unaccounted for in the interval $t \in [0, x_k]$. The key idea is to ``slide'' the excess blocks (corresponding to voters $v$ for which $i^*\cg_v k$) over by $x_k$ (see \Cref{fig:shifted-blocks}). 

\begin{figure}[h!] %
  \centering
  \includegraphics[width=0.5\textwidth]{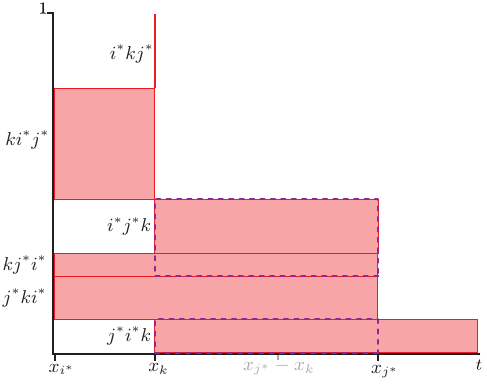} 
  \caption{After sliding the blocks corresponding to voters that prefer $i^*$ over $k$, the height of mass between $[0, x_k]$ is at least $s_{k\cg i^*}$, and the height of mass between $[x_k, x_{j^*}]$ is at least $s_{j^*\cg k}$ (highlighted by the dashed boxes).} 
  \label{fig:shifted-blocks}
\end{figure}

This change has the effect of $k$ becoming a proxy for $i^*$. For example, if we have a voter $v$ such that $j\cg_v k$ for some candidate $j$, then we claim that the block corresponding to $j$ covers the interval $[x_k, x_j]$. Note that this block has length at least $x_j - x_k$. If $k \cg_v i^*$ then $j \cg_v i^*$, and so the block has length at least $x_j$, and therefore covers the entire interval $[0, x_j]$. On the other hand if $k \cg_v i^*$, then ordinarily, the block would cover the interval $[0, x_j - x_k]$, but since it has been shifted over, it instead covers the interval $[x_k, x_j]$. The ultimate result is that at any $t \geq x_k$ and any $j \notin I_t$, $t$ is below all of the blocks corresponding to voters $v\in S_{j \cg k}$. 

More formally, let $f_v(t) = \displaystyle\oneif[0 \leq t \leq  \max_{i, j: j \cgeq_v i} (x_j - x_i)]$ (representing the block corresponding to $v$). Observe that
$$\int_0^\infty(1 - s_{\forall j\cg i,  x_j - x_i \leq t}) \dif t = \int_{0}^\infty \frac1n \sum_{v\in V} f_v(t) \dif t.$$
On the other hand, let 
$$g_v(t) = \begin{cases}
f_v(t)  & k \cg_v i^*\\
f_v(t - x_k)  & i^* \cg_v k\\
\end{cases}.$$
These functions correspond to the ``shifted'' blocks. Now observe that $\int_{0}^\infty f_v(t)\dif t = \int_{0}^\infty g_v(t) \dif t$ for all $v$, and so 
$$\int_0^\infty(1 - s_{\forall j\cg i,  x_j - x_i \leq t}) \dif t = \int_{0}^\infty \frac1n \sum_{v\in V} g_v(t) \dif t.$$
Observe that if $v \in S_{k\cg i^*}$, then $g_v(t) = f_v(t) = 1$ for all $t \in [0, x_k]$. Therefore,
$$\int_{0}^\infty \frac1n \sum_{v\in V} g_v(t) \dif t \geq x_k \cdot s_{k \cg i^*} + \int_{x_k}^\infty \frac1n \sum_{v\in V} g_v(t) \dif t.$$
(In fact, this is true with equality, since $g_v(t) = 0$ for  $v \in S_{i^*\cg k}$ and $t \in [0, x_k]$.) 

Now, we claim that for $t \geq x_k$, if $v \notin S_{k, i^* \cg J_t}$, then $g_v(t) = 1$. If $v \notin S_{k, i^* \cg J_t}$, then there exists some $j$ such that $x_j > t$ and either $j \cg_v i^*$ or $j \cg_v k$. We will consider two cases, depending on $v$'s preference between $k$ and $i^*$. 

Suppose that $k \cg_v i^*$. Then $g_v(t) = f_v(t)$ and in either case, we have that $j \cg_v i^*$. It follows that $f_v$ is 1 on the interval $[0, x_j]\ni t$, so $g_v(t) = 1$ as claimed. 

Now suppose that $i^* \cg_v k$. Then $g_v(t) = f_v(t - x_k)$.If $j \cg_v k$, then $f_v$ is $1$ on the interval $[0, x_j - x_k]$, and so $g_v$ is $1$ on the interval $[x_k, x_j] \ni t$.  If $j \cg_v i^*$, then again $f_v$ is $1$ on the interval $[0, x_j]$, which is even stronger than the previous case. Thus, $g_v(1) = 1$ as claimed. 

It follows that
$$\int_{x_k}^\infty \frac1n \sum_{v\in V} g_v(t) \dif t \geq \int_{x_k}^\infty (1 - s_{i^*, k \cg J_t}) \dif t .$$

Putting everything together, we have
$$\int_0^\infty(1 - s_{\forall j\cg i,  x_j - x_i \leq t}) \dif t  \geq x_k \cdot s_{k \cg i^*} + \int_{x_k}^\infty (1 - s_{i^*, k \cg J_t}) \dif t.$$
On the other hand, 
$$\int_0^\infty  s_{I_t \cg j^*} \dif t \leq x_k \cdot s_{i^*\cg j^*} + \int_{x_k}^\infty s_{I_t \cg j^*} \dif t.$$
Comparing terms, it is clear that the conditions of the lemma imply \cref{eq:precise-biased-det}, and so the result follows.
\end{hideproofAone}

\begin{remark}\label{rem:ub-hard}
As an aside, we mention one challenge in improving \Cref{lem:post-shift} and getting a stronger upper bound. A natural temptation would be to improve \Cref{lem:post-shift} similarly to how it improves \Cref{lem:lp-tournament}. We could introduce two candidates $k_1$ and $k_2$, or a sequence of candidates $k_1, k_2, \dots, k_r$. One could hope to show that after shifting the blocks that make up $2\SC(i^*)$, we can ensure that between the intervals $[x_{k_u}, x_{k_{u+1}}]$, the height of mass is at least $s_{k_{u+1}\cg k_u}$ (akin to \Cref{fig:shifted-blocks}). The issue is that with more than one intermediate candidate, a certain voter may be overextended across the intervals she must count towards. For example, if a voter has the preference $k_1 \cg i^* \cg k_2 \cg j^*$, then she must count towards $[x_{i^*}, x_{k_1}]$ and $[x_{k_2}, x_{j^*}]$, but the length of her block is only the maximum length of these intervals, not the sum. We note that if we could ignore this issue and suppose that the argument works, it would actually imply that \textit{Ranked Pairs} has distortion 3. The fact that the first step of the argument works is the reason why Ranked Pairs has distortion 3 when there are at most 4 candidates (see \Cref{thm:rp}), and the plausibility of this approach suggests why it was reasonable to conjecture that Ranked Pairs has distortion 3 in general. 
\end{remark}

We can use \Cref{lem:post-shift} to find sufficient conditions for a low distortion candidate based on local properties of the tournament graph.

\begin{corollary}\label{cor:local}
$\SC(j^*) \leq (1 + 2\lambda) \SC(i^*)$ if there exists a candidate $k$ such that  $k = i^*$ or $s_{k \cg i^*} \geq \frac1\lambda$, and one of the following conditions holds.
\begin{enumerate}[label=\textnormal{(\Roman*)}]

    \item\label{enum:1} $s_{k \cg j^*} \leq \frac{\lambda}{1 + \lambda}$.

    \item\label{enum:2} There exists a candidate $\ell \neq k$ such that  $ s_{k \cg j^*} \leq \lambda s_{\ell \cg k}$, and  $s_{\ell \cg j^*} \leq \lambda s_{j^* \cg k}$.
\end{enumerate}
\end{corollary}

\begin{figure}[h!] %
  \centering
  \includegraphics[width=0.9\textwidth]{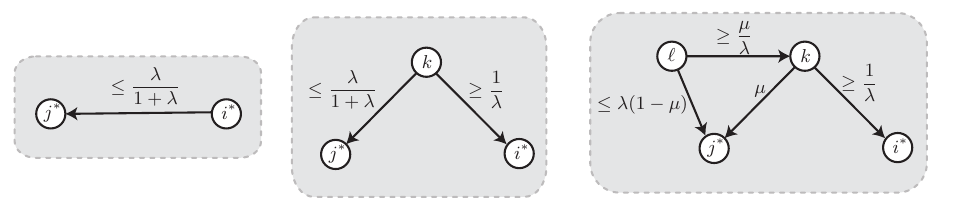} 
  \caption{Examples of local structures where \Cref{cor:local} shows that $\SC(j^*) \leq (1 + 2\lambda)\SC(i^*)$. } 
  \label{fig:local}
\end{figure}

\begin{hideproofAone}[Proof of \Cref{cor:local}]
We will prove that each of the conditions is sufficient one at at time, by showing that they imply the conditions of \Cref{cor:post-shift} (with the same choice of $k$).  First, note that $k = i^*$ or
$$s_{i^*\cg j^*} \leq 1 \leq \lambda s_{k \cg i^*}$$
so the first condition is satisfied. It remains to check the second condition given \ref{enum:1} and  \ref{enum:2}.

\ref{enum:1}: 
If $s_{k \cg j^*} \leq \frac{\lambda}{1 + \lambda}$ then $s_{j^* \cg k} \geq \frac{1}{1 + \lambda}$. Then for all partitions $I\sqcup J$ with $i^*, k^* \in I$, $j^*\in J$, we have
$$\min_{i\in I} s_{i\cg j^*} \leq s_{k \cg j^*} \leq \frac{\lambda}{1 + \lambda} \leq \lambda \cdot s_{j^*\cg k} \leq \lambda \max_{j\in J} \{s_{j\cg i^*},  s_{j\cg k}\}$$
as needed. 

\ref{enum:2}:
Consider a partition $I\sqcup J$ with $i^*, k^* \in I$, $j^*\in J$. If $\ell \in I$, then 
$$\min_{i\in I} s_{i\cg j^*} \leq s_{\ell \cg j^*} \leq \lambda s_{j^* \cg k} \leq \lambda\max_{j\in J} \{s_{j\cg i^*},  s_{j\cg k}\}.$$
If $\ell \in J$, then
\[
\min_{i\in I} s_{i\cg j^*} \leq s_{k \cg j^*} \leq \lambda s_{\ell \cg k} \leq \lambda \max_{j\in J} \{s_{j\cg i^*},  s_{j\cg k}\}. \qedhere
\]
\end{hideproofAone}

The first part of \Cref{cor:local} is the meat of the proof that \cite{DBLP:conf/ec/MunagalaW19}'s rule has distortion $2 + \sqrt{5}$, and is also a direct corollary of \cite[Corollary~5.1]{DBLP:conf/aaai/Kempe20a}). \cite{DBLP:conf/ec/MunagalaW19} defined the following weighted version of the well-known \emph{uncovered set}. 

\begin{definition}[covering]
The candidate $i^*$ \emph{$\beta$-covers} $j^*$ if $s_{i^* \cg j^*} > \beta$ and for all $k$ such that $s_{k\cg i^*} \geq \beta$, we have $s_{k \cg j^*} > \beta$. Conversely, $j^*$ is \emph{$\beta$-uncovered} if for all $i^*\neq j^*$, there exists a candidate $k$ such that $k = i^*$ or $s_{k\cg i^*} \geq \beta$, and  $s_{k \cg j^*} \leq \beta$.
\end{definition}

\cite[Lemma~3.2]{DBLP:conf/ec/MunagalaW19} showed that for any $\beta \in [\frac12, 1]$ there is a candidate that is $\beta$-uncovered. Though their proof was indirect, it is not hard to see that
$$\arg\min_{j\in C} |\{k \in C: s_{k\cg j} > \beta\}|$$
must be $\beta$-uncovered, since if $i^*$ covers $j^*$ then $\{k \in C: s_{k\cg j^*} > \beta\} \supseteq \{i^*\} \cup \{k \in C: s_{k\cg i^*} > \beta\}.$ So in fact, a candidate in the $\beta$-uncovered set can be found with a weighted version of \emph{Copeland's rule} (which chooses the candidate that beats the most candidates in a majority vote).  

Now, we can see that if $j^*$ is a $\beta$-uncovered candidate, then it satisfies \Cref{cor:local} via \ref{enum:1}, so long as $ \frac{1}{\lambda} \leq \beta \leq \frac{\lambda}{1 + \lambda}$. Therefore, choosing $\lambda = \frac{1 + \sqrt{5}}{2}$ (which satisfies $\frac{1}{\lambda} = \frac{\lambda}{1 + \lambda}$), there exists a candidate with distortion at most $1 + 2\lambda = 2 + \sqrt{5}$.  

However, by folding in \ref{enum:2}, we can improve this to $1 + 2\lambda$ where $\lambda$ is the real root of $\lambda^3 - \lambda^2 - 1$.

%% file: sections/determ_two_ub.tex
\NewEnviron{hideproofAtwo}[1][Proof]{%
  \ifshowproofs
    \begin{proof}[#1]
      \BODY
    \end{proof}
  \else
    \par\noindent\textit{#1.} See \Cref{prf:tourn-ub}. \hfill$\square$%
  \par\nobreak\vskip\medskipamount
  \fi
}

\subsection{Upper Bound: \texorpdfstring{\nameref{box:unblanketed}}{Unblanketed Set}}\label{sec:tourn-ub}

The key idea is to create a stronger version of the covering condition, called ``blanketing'', built using the conditions of \Cref{cor:local}. We take a similar approach to \cite{DBLP:conf/ec/MunagalaW19} to prove that there always exists a candidate that is unblanketed: We suppose that there exists some cycle where each candidate in the cycle blankets its neighbor, and then derive a contradiction. In our case, due to the added complexity in the condition, we need to construct a second cycle using the first in order to derive a contradiction.

\begin{definition}[blanketing]
The candidate $i^*$ \textit{$(\alpha, \beta)$-blankets} $j^*$ 
if for all candidates $k$ such that either $k = i^*$ or $s_{k \cg i^*} \geq \alpha$, all of the following hold:
\begin{enumerate}[label=\textnormal{(\Roman*)}] 
    \item $s_{k \cg j^*} > \beta$.
    \item $s_{k\cg j^*} > \alpha$ or for all $\ell \neq k, j^*$, if $s_{k \cg \ell} \leq \beta$ then $s_{j^* \cg \ell} < \beta$.
\end{enumerate}

Conversely, $j^*$ is \textit{$(\alpha, \beta)$-unblanketed} if for all $i^* \neq j^*$, there exists a candidate $k \in C$ such that either $k = i^*$ or $s_{k \cg i^*} \geq \alpha$, and one of the following conditions holds.
\begin{enumerate}[label=\textnormal{(\Roman*)}]  
    \item $s_{k\cg j^*} \leq \beta$.
    \item $s_{k\cg j^*} \leq \alpha$ and there exists $\ell \neq k, j^*$ such that $s_{k\cg \ell} \leq \beta \leq s_{j^*\cg \ell}$.
\end{enumerate}
\end{definition}

\begin{mybox}[label={box:unblanketed},nameref={Unblanketed Set}]{The \underline{Unblanketed Set} Voting Rule (with Parameters $\alpha \geq \beta > \frac12$).}
\begin{itemize}
\item[$\triangleright$] Let the winner be any $j^* \in C$ that is $(\alpha, \beta)$-unblanketed.
\end{itemize}
\end{mybox}

\begin{lemma}\label{lem:ubk}
For all $\alpha \geq \beta > \frac12$, there exists a candidate $j^*$ that is $(\alpha, \beta)$-unblanketed.
\end{lemma}

\begin{hideproofAtwo}[Proof of \Cref{lem:ubk}]
Suppose towards a contradiction that we have an election such that every candidate is $(\alpha, \beta)$-blanketed by some other candidate. 

Then we can construct a cycle $\mathcal{C}$ of candidates such that for each edge $(i, j) \in \mathcal{C}$, $i$ $(\alpha, \beta)$-blankets $j$. (See the left of \Cref{fig:cycles}) In particular, this implies that for each such edge, $s_{i \cg j} > \beta$. 

\begin{figure}[h!] %
  \centering
  \includegraphics[width=0.9\textwidth]{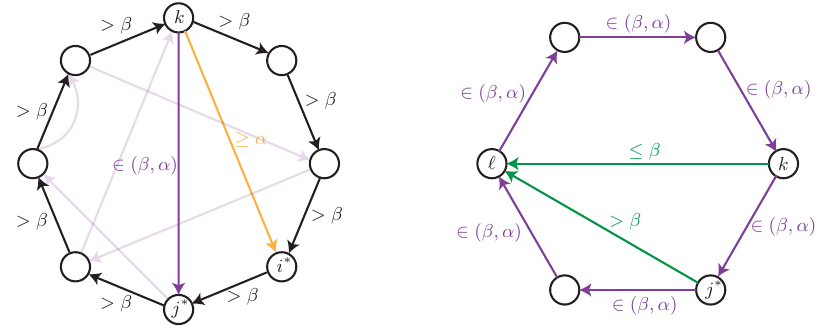} 
  \caption{Left: the cycle $\mathcal{C}$. Right: the cycle $\hat{\mathcal{C}}$. We derive a contradiction by showing that some $i^*$ cannot blanket some $j^*$, using $k$ and $\ell$ as labeled.} 
  \label{fig:cycles}
\end{figure}

Fix a candidate $k$ in the cycle. We claim that there exists some edge $(i^*, j^*) \in \mathcal{C}$ such that  $s_{k\cg j^*} < \alpha$ and either $k = i^*$ or $s_{k \cg i^*} \geq \alpha$. If $(k', k) \in \mathcal{C}$, then $s_{k\cg k'} < 1 - \beta < \alpha$. Therefore, there must be some candidate $j^*$ whose distance from $k$ along the cycle is minimal, and for which $s_{k \cg j^*} < \alpha$. The minimality of $j^*$ implies that either $(k, j^*)\in \mathcal{C}$,  or if $(i^*, j^*) \in \mathcal{C}$, then $s_{k \cg i^*} \geq \alpha$ as claimed.

Since $i^*$ $(\alpha, \beta)$-blankets $j^*$, it follows that $s_{k \cg j^*} > \beta$. Now, we can iterate this analysis, repeatedly replacing $j^*$ with $k$, to construct a second cycle of candidates $ \hat{\mathcal{C}}$ (shown on the right of \Cref{fig:cycles}).

Now fix a candidate $\ell$ in $ \hat{\mathcal{C}}$. By a similar argument to the above, we claim that there exists an edge $(k, j^*) \in \hat{\mathcal{C}}$ such that $\ell \neq k, j^*$ and $s_{k \cg \ell} \leq \beta < s_{j^* \cg \ell}.$ If $(\ell', \ell) \in \hat{\mathcal{C}}$ then $s_{\ell' \cg \ell} > \beta$. Therefore there must be some candidate $j^*$ whose distance from $\ell$ along $\hat{\mathcal{C}}$ is minimal for which $s_{j^* \cg \ell} > \beta$. Note that $(\ell, j^*) \notin \hat{\mathcal{C}}$, otherwise $s_{j^* \cg \ell} < 1 - \beta < \beta$. Therefore, if $(k, j^*) \in  \hat{\mathcal{C}}$, then $\ell \neq k, j^*$, and by the minimality of $j^*$, $s_{k \cg \ell} \leq \beta$ as claimed. 

The existence of $\ell$ violates the condition that $i^*$ blankets $j^*$, and so we have a contradiction.
\end{hideproofAtwo}

\begin{theorem}\label{thm:2ub}
Let $\lambda$ be the real root of $x^3 - x^2 - 1$. The distortion of \nameref{box:unblanketed} is at most $1 + 2\lambda \approx 3.93114$ with parameters $\alpha = \frac{1}{\lambda}\approx 0.68233$  and $\beta = 1 - \frac{1}{\lambda^2} = 2 - \lambda \approx 0.53443$.
\end{theorem}
\begin{hideproofAtwo}[Proof of \Cref{thm:2ub}]
Intuitively, looking at the rightmost condition in \Cref{fig:local}, by setting $\alpha = \frac{1}{\lambda}$, we can get $\mu \leq \frac{1}{\lambda}$, and by setting $\frac{\mu}{\lambda} = \lambda (1 - \mu)$ (i.e., $\frac{1}{\lambda^2} = \lambda - 1$), we can get the required condition on $\ell$ with $\beta = 1 - \frac{1}{\lambda^2} = 2 - \lambda$.

We will arrive at the theorem organically, by determining the best value of $\lambda$ for which an $(\alpha, \beta)$-unblanketed candidate has distortion at most $1 + 2\lambda$. 

Suppose we have an $(\alpha, \beta)$-unblanketed candidate $j^*$, which we would like to satisfy the conditions of \Cref{cor:local} for the smallest value of $\lambda$. By the initial condition of the rule, there exists $k$ such that $k = i^*$ or $s_{k\cg i^*} \geq \frac1\lambda$, so long as $\alpha \geq \frac1\lambda$.

Next, to satisfy \ref{enum:1}, it suffices if $\beta \leq \frac{\lambda}{1 + \lambda}$. And finally, to satisfy \ref{enum:2} it suffices if $ \lambda(1-\beta) \geq \alpha$ and $1 - \beta \leq \lambda(1 - \alpha)$, since then
$$s_{k\cg j^*} \leq \alpha \leq \lambda(1-\beta) \leq \lambda s_{\ell\cg k}$$
and
$$s_{\ell\cg j^*} \leq 1 - \beta \leq \lambda( 1- \alpha) \leq \lambda s_{j^*\cg k}.$$
Combining our constraints, by \Cref{cor:local}, $j^*$ has distortion at most $1 + 2 \lambda$ if 
$$\lambda \geq \max\left(\frac1\alpha, \frac{\beta}{1 - \beta}, \frac{\alpha}{1 - \beta}, \frac{1 - \beta}{1 - \alpha}\right).$$
It is not hard to see that the expression on the right is minimized with the claimed choice of parameters.
\end{hideproofAtwo}

%% file: sections/determ_two_lb.tex
\NewEnviron{hideproofAthree}[1][Proof]{%
  \ifshowproofs
    \begin{proof}[#1]
      \BODY
    \end{proof}
  \else
      \par\noindent\textit{#1.} See \Cref{prf:tourn-lb}. \hfill$\square$%
  \par\nobreak\vskip\medskipamount
  \fi
}

\subsection{Universal Lower Bound}\label{sec:tourn-lb}

In this section, we will prove a universal lower bound on the distortion of any deterministic tournament rule, using an election with 5 candidates. Since \cite{DBLP:journals/ai/AnshelevichBEPS18} proved the ranked pairs has distortion at most 3 when there are at most 4 candidates (see \Cref{thm:rp}), our instance is as small as possible. 

Before diving into the details, we offer some intuition for the construction. Our goal is to create a cycle of candidates such that each candidate in the cycle can have large distortion with respect to the candidate preceding it in the cycle. This is the outside cycle in \Cref{fig:5lb}. Observing the way that conditions such as \Cref{lem:lp-tournament,lem:post-shift} bound $\lambda$ by ratios between edges in the tournament graph, it is natural for edges along this cycle to increase exponentially with a rate of $\lambda$. To show a lower bound on the distortion, we can use \Cref{thm:biased-iff}. We simply need to exhibit a biased metric for which \cref{eq:precise-biased-det} fails. 

For the exponentially increasing edges, we can use the same simple biased metrics used by \cite{DBLP:conf/soda/CharikarR22} to prove a lower bound against randomized voting rules (called $(0,1,2,3)$-metrics, for the different distances in these metrics after scaling). These correspond to metrics where $x_{i^*} = 0$ and $x_{j} = 1$ for $j \neq i^*$. For the smallest edge in the cycle we use a different metric---a ``half-integral'' analogue of the $(0,1,2,3)$-metrics, where for each $j \neq i^*$,  $x_j \in \{\frac12, 1\}$.

\begin{theorem}\label{thm:lb} 
The distortion of any deterministic \textnormal{($2$-)}tournament rule is at least $3.11287$.
\end{theorem}

\begin{hideproofAthree}[Proof of \Cref{thm:lb}]
The goal is to construct a tournament graph with 5 candidates, such that for each candidate $j^*$, there exists some underlying metric space, and realization of the preferences in terms of full rankings, such that the distortion of $j^*$ is as claimed.

We will start by constructing the tournament graph in terms of two general parameters $\beta$ and $\lambda$, and then explain what conditions are necessary for this tournament graph to give us a general lower bound. We will eventually set $\beta \approx 0.60696$, $\lambda \approx 1.056439$,  and our eventual distortion will be $1 + 2\lambda \approx 3.11287$.

Our candidates are labeled $0,1,2,3,4$, treated modulo 5. The comparisons matrix, whose $(i, j)$ entry is $s_{i\cg j}$ is as follows. The tournament graph is visualized in \Cref{fig:5lb}.
\[
\begin{bmatrix}
  - & \beta\lambda & 1-\beta & 1-\beta  & 1-\beta  \\
  1-\beta\lambda & - & \beta \lambda^2 & 1-\beta  & 1-\beta  \\
  \beta & 1-\beta\lambda^2 &- & \beta \lambda^3 & \beta\lambda^3(1 + \lambda) -1 \\
  \beta & \beta & 1-\beta\lambda^3 & - & \beta \lambda^4 \\
  \beta & \beta & 2 - \beta\lambda^3(1 + \lambda) & 1-\beta\lambda^4 & -
\end{bmatrix} \approx \begin{bmatrix}
- & 0.641 & 0.393 & 0.393 & 0.393 \\
0.359 & - & 0.677 & 0.393 & 0.393 \\
0.607 & 0.323 & - & 0.716 & 0.472 \\
0.607 & 0.607 & 0.284 & - & 0.756 \\
0.607 & 0.607 & 0.528 & 0.244 & -
\end{bmatrix}
\]

\begin{figure}[h!] %
  \centering
  \includegraphics[width=0.5\textwidth]{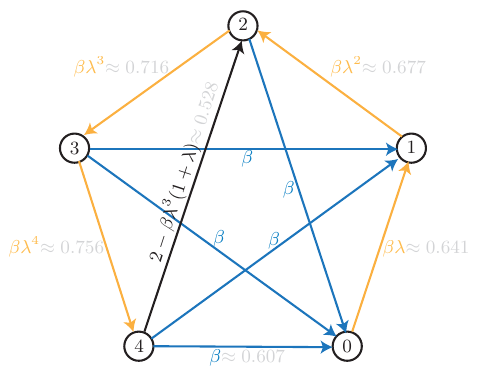} 
  \caption{A tournament graph where each candidate has worst-case distortion at least 3.112.} 
  \label{fig:5lb}
\end{figure}

At this point, we note that a necessary (but insufficient) condition for a tournament graph to be realizable with rankings is that for all candidates $i, j, k$, $s_{i\cg j} + s_{j\cg k} + s_{k \cg i} \leq 2$, since each voter counts towards at most two of the three terms. The triple of candidates $(3, 1, 2)$ will therefore require that
$$\beta + \beta\lambda^2 + \beta\lambda^3 \leq 2$$
and so 
$$\beta \leq \frac{2}{1 + \lambda^2 + \lambda^3}.$$
In fact, this inequality will be satisfied with equality, but we keep $\beta$ as is for ease of notation.

Using \Cref{thm:biased-iff}, we can show that each candidate $j^*$ has large worst-case distortion by exhibiting a biased metric where some other candidate $i^*$ is optimal, and $\SC(j^*) - \SC(i^*) \geq \lambda \cdot 2\SC(i^*)$ (as given by \cref{eq:precise-biased-det}). For each $j^*$, $i^*$ will be the candidate labeled $j^* - 1$ (modulo 5). 

For $j^* = 1,2,3,4$, these metrics are relatively straightforward. Each is defined by the biased metric where $x_{j^* - 1} = 0$, and $x_k = 1$ for $k \neq j^* - 1$. (These are precisely the \textit{$(0,1,2,3)$-metrics} used by \cite{DBLP:conf/soda/CharikarR22} to give nontrivial lower bounds on the distortion of randomized voting rules.) For these metrics, we will have that 
$$\SC(j^*) - \SC(j^* - 1) =  \int_0^\infty  s_{I_t \cg j^*}  \dif t = s_{j^* - 1 \cg j^*} = \beta\lambda^{j^*}.$$
and
$$2\SC(j^* - 1) = \int_0^\infty(1 - s_{\forall j\cg i,  x_j - x_i \leq t}) \dif t = 1 - \plu(j^* - 1)$$
where $\plu(j)$ denotes the fraction of voters that rank $j$ first. Note that
$$1 - \plu(j^* - 1) \geq \max_{k \neq j^* - 1} s_{k \cg j^* - 1} = s_{j^* - 2 \cg j^* - 1} = \beta\lambda^{j^* - 1}.$$
Therefore, if we can construct the full preference profile so that $1 - \plu(j^* - 1) = s_{j^* - 2 \cg j^* - 1}$ (i.e., if $j^* - 1 \cg_v j^* - 2$, then $v$ ranks $j^* - 1$ first), then this shows that $\SC(j^*) \geq (1 + 2\lambda)\SC(j^* - 1)$. 

As it turns out, the condition that $1 - \plu(j^* - 1) = s_{j^* - 2 \cg j^* - 1}$ is easy to satisfy. One can imagine a greedy approach, first breaking up the voters by their preference over $j^* - 1$ and $j^* - 2$, and then inserting the candidates appropriately (in terms of position and proportion) to satisfy each of the tournament constraints. To save a considerable amount of messy algebra, we simply provide the optimal preference profiles for the final setting of parameters in \Cref{tab:lb1234}.

\begin{table}[h!]
    \centering
    \begin{minipage}{0.24\textwidth}
        \centering
        \begin{tabular}{c|c}
            $\sigma$ & $s_\sigma$ \\ \hline
01423 & 0.03426\\
03412 & 0.28436\\
04231 & 0.07442\\
\hline
12340 & 0.35879\\
23401 & 0.11288\\
42301 & 0.13530
        \end{tabular}
    \end{minipage}%
    \begin{minipage}{0.24\textwidth}
        \centering
        \begin{tabular}{c|c}
            $\sigma$ & $s_\sigma$ \\ \hline
12304 & 0.17939\\
14203 & 0.17939\\\hline 
02341 & 0.07442\\
20341 & 0.10497\\
23014 & 0.03426\\
23401 & 0.07862\\
34012 & 0.28436\\
40123 & 0.03426\\
42301 & 0.03033
        \end{tabular}
    \end{minipage}%
    \begin{minipage}{0.24\textwidth}
        \centering
        \begin{tabular}{c|c}
            $\sigma$ & $s_\sigma$ \\ \hline
23401 & 0.32260\\ \hline
01234 & 0.03426\\
03412 & 0.28436\\
12304 & 0.07442\\
12340 & 0.04039\\
14023 & 0.07442\\
14230 & 0.16955
        \end{tabular}
    \end{minipage}%
    \begin{minipage}{0.24\textwidth}
        \centering
        \begin{tabular}{c|c}
            $\sigma$ & $s_\sigma$ \\ \hline
30142 & 0.03426\\
31420 & 0.10497\\
34012 & 0.14514\\ \hline
12304 & 0.17939\\
20341 & 0.17939\\
23140 & 0.07442\\
23401 & 0.03845\\
40123 & 0.21365\\
42301 & 0.03033
        \end{tabular}
    \end{minipage}
    \caption{The underlying preference profiles used for our lower bounds against candidates 1,2,3,4. From left to right, the optimal candidates are 0,1,2,3. The voters that rank the optimal candidate $j^* - 1$ first are above the bar. The remaining voters all prefer $j^* - 2$ over $j^* - 1$.}\label{tab:lb1234}
\end{table}

It remains to show the candidate 0 can have large cost compared to candidate 4. Consider the biased metric with $x_0 = x_1 = 1$, $x_2 = x_3 = \frac12$, $x_4 = 0.$ Then we have
$$\SC(0) - \SC(4) =  \int_0^\infty  s_{I_t \cg j^*}  \dif t = \tfrac12 s_{4 \cg 0} + \tfrac12 s_{2,3,4 \cg 0} \leq \beta.$$
We can satisfy the above with equality if we construct a preference profile such that if a voter $v$ prefers $0$ over any of $2,3,4$, then $0 \cg 2,3,4$ (and so $s_{2\cg0} = s_{3\cg0} = s_{4\cg0} = s_{2,3,4\cg0} = \beta$). On the other hand,
\begin{align*}
2\SC(4) = \int_0^\infty(1 - s_{\forall j\cg i,  x_j - x_i \leq t}) \dif t &= \tfrac12 (1 - \plu(4)) + \tfrac12 (1 - s_{4 \cg 0,1})\\
&\geq \tfrac12 s_{3\cg 4} + \tfrac12 \max(s_{0\cg 4}, s_{1\cg 4}) = \tfrac12 \beta\lambda^4 + \tfrac12(1 - \beta).
\end{align*}
Similarly, we can satisfy the above with equality if we have a preference profile where $3\cg_v 4$ whenever $4$ is not $v$'s top choice, and  $v$ always ranks both or neither of $0, 1$ above $4$.

All of these conditions are indeed satisfied by the preference profile given in \Cref{tab:lb0}. 

\begin{table}[ht!]
\centering
\begin{tabular}{c|c}
$\sigma$ & $s_\sigma$ \\ \hline
01234 & 0.31862\\
10234 & 0.07442\\
23401 & 0.07862\\
34120 & 0.28436\\
42301 & 0.24397
\end{tabular}
\caption{The underlying preference profiles used for our lower bounds against candidate 0. The optimal candidate is 4.}\label{tab:lb0}
\end{table}

In this case, we have
$$\frac{\SC(0) - \SC(4)}{2\SC(4)} = \frac{\beta}{\tfrac12(\beta\lambda^4 + 1 - \beta)},$$
so to get distortion $1 + 2\lambda$, we would like the expression above to be equal to $\lambda$. With $\beta = \frac{2}{1 + \lambda^2 + \lambda^3}$, we find that 
$$2\lambda^5 + \lambda^4 + \lambda^3 - \lambda - 4 = 0.$$
The solution to this equation is $\lambda \approx 1.056439$ as claimed.
\end{hideproofAthree}

\begin{remark}\label{rem:lb-hard}
We note that a natural temptation to get a stronger result would be to attempt the same approach with larger cycles. The challenge however, is that with exponentially increasing weights along the cycle, the weights in the tournament graph can be significantly constrained. We begin to see this effect in \Cref{fig:5lb}, where $s_{4\cg2}$ is forced to be small because $s_{2\cg 3}$ and $s_{3\cg 4}$ are large. As a result, our lower bound benefits from the cycle being as small as possible.
\end{remark}

%% file: sections/k-selection.tex
\section{Selection in \texorpdfstring{$k$}{k}-Tournaments}\label{sec:k_tournament}

We start with some preliminaries on Stable Lotteries, then introduce \nameref{lem:main}, and finally give two new rules: \nameref{box:simul_veto} and \nameref{box:two_lotteries}. Both of these rules are $(k+1)$-tournament rules when using stable $k$-lotteries as a component.

%% file: sections/stable-lotteries.tex
\NewEnviron{hideproofBone}[1][Proof]{%
  \ifshowproofs
    \begin{proof}[#1]
      \BODY
    \end{proof}
  \else
    \par\noindent\textit{#1.} See \Cref{prf:stable}. \hfill$\square$%
    \par\nobreak\vskip\medskipamount
  \fi
}

\subsection{Stable Lotteries}
\label{subsec:prelim/stable}
We first extend the preference of each voter from comparisons between candidates to comparisons between multisets of candidates. Our definition is that between two multisets, the voter prefers the multiset that contains her favorite candidate (favorite within the two multisets). Tie-breaking is crucial here, and is defined in the following way. Each candidate copy in each multiset is attached with a freshly drawn number from a uniform distribution on $[0, 1]$. When tied, the copy with the highest attachment is the most preferred. More precisely, given multisets $S,T$ of candidates, suppose candidate $c$ is voter $v$'s highest ranked candidate in $S \cup T$. Let $n(S)$ (respectively, $n(T)$) be the number of copies of $c$ in $S$ (respectively, $T$). Then voter $v$ prefers $S$ to $T$ with probability $\frac{n(S)}{n(S)+n(T)}$. For example, any voter prefers the multiset $\{c, c, c\}$ to the multiset $\{c, c\}$ with probability $\frac{3}{5}$. We use this tie-breaking rule so that in later discussion on stable $k$-lotteries, the value of a zero-sum game is exactly $1 - \frac{1}{k + 1}$, which is a helpful property used in the proof of \nameref{lem:main}.

For any distribution $D$ over candidates, we denote by $D^k$ the distribution of a size-$k$ multi-set obtained by drawing from $D$ i.i.d.\@ $k$ times.

\begin{definition}[stable $k$-lotteries]
\label{def:stable}
We call a distribution $D$ on the candidates a \emph{stable $k$-lottery} ($k \in \mathbb{Z}^+$) if for any distribution $D'$ on the candidates, it holds that
\[
\Pr_{v \sim V}[D^k \succ_v D'] \geq 1 - \frac{1}{k + 1}.
\]
We call a distribution $D'$ on the candidates a \emph{stable defensive $k$-lottery}, if for any distribution $D$,
\[
\Pr_{v \sim V}[D^k \succ_v D'] \leq 1 - \frac{1}{k + 1}.
\]
\end{definition}
Stable $k$-lotteries generalize Maximal Lotteries \cite{kreweras1965aggregation,fishburn1984probabilistic}, which precisely correspond to stable $1$-lotteries (and stable defensive $1$-lotteries). We comment that the work of \cite{DBLP:journals/teco/ChengJMW20} proves that there is always a distribution over size-$k$ sets of candidates satisfying the condition in \cref{def:stable}. However, their distribution might not come from by i.i.d.\@ sampling from the same base distribution. These two definitions are sometimes different, as illustrated by \cref{ex:stable} below, but the worst-case guarantees of $1 - \frac{1}{k + 1}$ remain the same. The fact that stable lotteries can be implemented with independent sampling is not explicit in \cite{DBLP:journals/teco/ChengJMW20}, but it is a relatively straightforward consequence of the proof via the minimax theorem. We include a proof of \cref{thm:existence_lottery} for completeness.

\begin{example}
\label{ex:stable}
Consider a profile where $n = m$ is a large even number, and the voters' preferences are cyclic: each voter $v_i$ has preference $c_i \succ c_{i + 1} \succ \cdots \succ c_m \succ c_1 \succ c_2 \succ \cdots \succ c_{i - 1}$. The distribution that picks $\{c_i, c_{i + \frac{m}{2}}\}$ for a uniform $i \in \{1, 2, \ldots, \frac{m}{2}\}$ beats any distribution over candidates with probability about $\frac{3}{4}$, but two i.i.d.\@ samples from any distribution over candidates can only guarantee $\frac{2}{3}$ because of \cref{obs:lottery_sym}.
\end{example}

\cref{obs:lottery_sym} directly follows from symmetry, and it shows that the guarantee of $1 - \frac{1}{k + 1}$ for stable $k$-lotteries is tight by taking $D' = D$ in \cref{def:stable}.
\begin{observation}
\label{obs:lottery_sym}
For any distribution $D$ on the candidates,
\[
\Pr_{v \sim V}[D^k \succ_v D] = 1 - \frac{1}{k + 1}.
\]
\end{observation}

Although we have defined both stable lotteries and stable defensive lotteries, it turns out that every stable $k$-lottery is also a stable defensive $k$-lottery. To the best of our knowledge, this fact was unknown prior to our work.
\begin{theorem}[stable $k$-lotteries are stable defensive $k$-lotteries]
\label{thm:stable_is_defensive_stable}
For any $k$ and any preference profile, any stable $k$-lottery $D$ is also a stable defensive $k$-lottery.
\end{theorem}
\begin{hideproofBone}[Proof of \Cref{thm:stable_is_defensive_stable}]
For any distribution $A$ on the voters, we wish to prove
\[
\Pr_{v \sim V}[D \succ_v A^k] \geq \frac{1}{k + 1}.
\]
Slightly abusing the notation, let $V$ be the measure on the voters. Let $D_v$ (or $A_v$) be the distribution of the rank of the candidate drawn from $D$ (or $A$) according to the voter $v$. (Lower is better.) We can rewrite the inequality above that we wish to prove as
\[
\int_v \left(\int D_v \d A_v^k\right) \d V \leq \frac{k}{k + 1}.
\]

Since $D$ is a stable $k$-lottery, we have
\[
\int_v \left(\int D_v^k \d A_v\right) \d V \leq \frac{1}{k + 1}.
\]
Also, it holds that
\[
\int A_v^k \d A_v = \frac{1}{k + 1}.
\]

By H\"older's inequality, we have
\begin{align*}
\int_v \left(\int D_v A_v^{k - 1} \d A_v\right) \d V &\leq \int_v \left(\left(\int D_v^k \d A_v\right)^{\frac{1}{k}} \cdot \left(\int A_v^k \d A_v\right)^{\frac{k - 1}{k}}\right) \d V\\
&= \left(\frac{1}{k + 1}\right)^{\frac{k - 1}{k}} \cdot \int_v \left(\int D_v^k \d A_v\right)^{\frac{1}{k}} \d V.
\end{align*}
Further, by Jensen's inequality and the concavity of $x^{\frac{1}{k}}$, we get
\[
\int_v \left(\int D_v^k \d A_v\right)^{\frac{1}{k}} \d V \leq \left(\int_v \left(\int D_v^k \d A_v\right) \d V\right)^{\frac{1}{k}}.
\]

Putting them together, we have
\begin{align*}
\int_v \left(\int D_v \d A_v^k\right) \d V &= \int_v \left(\int k D_v A_v^{k - 1} \d A_v\right) \d V\\
&\leq k \cdot \left(\frac{1}{k + 1}\right)^{\frac{k - 1}{k}} \cdot \int_v \left(\int D_v^k \d A_v\right)^{\frac{1}{k}} \d V\\
&\leq k \cdot \left(\frac{1}{k + 1}\right)^{\frac{k - 1}{k}} \cdot \left(\int_v \left(\int D_v^k \d A_v\right) \d V\right)^{\frac{1}{k}}\\
&\leq k \cdot \left(\frac{1}{k + 1}\right)^{\frac{k - 1}{k}} \cdot \left(\frac{1}{k + 1}\right)^{\frac{1}{k}}\\
&= \frac{k}{k + 1},
\end{align*}
and therefore proves the statement.
\end{hideproofBone}

\begin{theorem}[existence of stable $k$-lotteries]
\label{thm:existence_lottery}
For any $k$ and any preference profile, a stable $k$-lottery always exists.
\end{theorem}
\begin{hideproofBone}[Proof of \Cref{thm:existence_lottery}]
Following \cite{DBLP:journals/teco/ChengJMW20,jiang2020approximately}, we prove the theorem using the minimax theorem \cite{v1928theorie}. 
Consider the zero-sum game between two players $A$ and $B$, where player $A$ plays a distribution $D$ on the candidates and player $B$ plays a distribution $D'$. The value of the game is $\Pr_{v \sim V}[D^k \succ_v D']$. Player $A$ wishes to maximize the value, while Player $B$ wishes to minimize the value.

First, we show that there is a pure-strategy Nash equilibrium.
On the one hand, for any distribution $\Delta$ over distributions $D$, if we set $D^*=\E[\Delta]$, then for any distribution $D'$ over candidates and for any fixed $v \in V$,
\begin{equation}
\label{eqn: concave}
\Pr_{D^* = \E[\Delta]}[({D^*})^k \succ_v D'] \geq \Pr_{D \sim \Delta}[D^k \succ_v D'].    
\end{equation}
This is because the left-hand side is
\[
\Pr_{D^* = \E[\Delta]}[({D^*})^k \succ_v D'] = 1 - \left(1 - \E_{D \sim \Delta}[\Pr[D \succ_v D']]\right)^k
\]
and the right-hand side is
\[
\Pr_{D \sim \Delta}[D^k \succ_v D'] = \E_{D \sim \Delta}\left[1 - \left(1 - \Pr[D \succ_v D']\right)^k\right].
\]
\cref{eqn: concave} follows from applying Jensen's inequality on the concave function $f(x) = 1 - (1 - x)^k$.
Therefore, against any mixed strategy $\Delta'$ on $D'$, the pure strategy $D^*$ achieves at least the same value as $\Delta$.

On the other hand, for each distribution $\Delta'$ over distributions $D'$, if we set $D^{**}=\E[\Delta']$,
then for any distribution $D$ and any fixed $v \in V$, by definition, we have
\[
\Pr[{D}^k \succ_v D^{**}] = \Pr_{D' \sim \Delta'}[D^k \succ_v D'].
\]
Therefore, against any mixed strategy $\Delta$ on $D$, the pure strategy $D^{**}$ achieves at most the same value as $\Delta'$.
Since there is always a mixed-strategy Nash equilibrium in this zero-sum game, from our discussion above, we conclude that there is always a pure-strategy Nash equilibrium.

Let $(D^*,D^{**})$ be a pure-strategy Nash equilibrium. From \Cref{obs:lottery_sym}, we have
\[
1 - \frac{1}{k + 1} = \Pr_{v \sim V}[(D^{**})^k \succ_v D^{**}] \leq \Pr_{v \sim V}[(D^*)^k \succ_v D^{**}] \leq \Pr_{v \sim V}[(D^{*})^k \succ_v D^{*}] = 1 - \frac{1}{k + 1}.
\]
Therefore, the value of the game must be $1-\frac{1}{k+1}$, and hence by definition, $D^*$ is a stable $k$-lottery, and $D^{**}$ is a stable defensive $k$-lottery.
\end{hideproofBone}

%% file: sections/main_lemma.tex
\NewEnviron{hideproofBtwo}[1][Proof]{%
  \ifshowproofs
    \begin{proof}[#1]
      \BODY
    \end{proof}
  \else
    \par\noindent\textit{#1.} See \Cref{prf:sr-lemma}. \hfill$\square$%
  \par\nobreak\vskip\medskipamount
  \fi
}

\subsection{The Stability-Representation Lemma}\label{sec:sr-lemma}

In this section, we introduce \nameref{lem:main}. It quantitatively characterizes one property of stable $k$-lotteries: if many voters think that a single candidate $i \in C \setminus J$ is better than all candidates in $J$, then a stable $k$-lottery cannot pick a candidate in $J$ too often; in other words, the subset $C \setminus J$ must be well-represented in the stable $k$-lottery.

\begin{lemma}[the stability-representation lemma]
\label{lem:main}
Fix any $J \subsetneq C$ and $i \in C \setminus J$. For a stable $k$-lottery $D$ with probability of $p_J$ to pick a candidate in $J$, it holds that for some $g(k) = \Theta\Big(\sqrt{\frac{\log k}{k}}\Big)$,
\[
p_J - g(k) \ \leq \ 1 - s_{i \succ J}.
\]
\end{lemma}
\begin{hideproofBtwo}[Proof of \Cref{lem:main}]
Let $S$ be the subset of voters with the preference $i \succ J$. Let $\alpha = s_{i \succ J} = \frac{|S|}{n}$ and $\lambda = p_J$.

Let $D^*$ be a stable $k$-lottery and hence a stable defensive $k'$-lottery (\cref{def:stable}, \cref{thm:stable_is_defensive_stable}) on the set $S$, with $k' = \Big\lceil(1 - \lambda)k + 2\sqrt{k} + \sqrt{k\ln k}\Big\rceil$. Since $D$ is a stable $k$-lottery, by definition, it holds that
\[
\Pr_{v \sim V} [D^k \succ_v D^*] \geq 1 - \frac{1}{k + 1}.
\]
This inequality and the fact that
\begin{align*}
\Pr_{v \sim V} [D^k \succ_v D^*] &= \alpha \cdot \Pr_{v \sim S} [D^k \succ_v D^*] + (1 - \alpha) \cdot \Pr_{v \sim (V \setminus S)} [D^k \succ_v D^*]\\
&\leq \alpha \cdot \Pr_{v \sim S} [D^k \succ_v D^*] + 1 - \alpha
\end{align*}
together imply
\[
\Pr_{v \sim S} [D^k \succ_v D^*] \geq 1 - \frac{1}{\alpha (k + 1)}.
\]

Denote $I=C\setminus J$, and let $D|_{I}$ be the conditional distribution of $D$ on the subset $I$ (that is, for each $t\in I$, $\Pr_{D|_{I}}[t]=\Pr_{D}[t\mid t\in I]$).
We define another distribution $\hat D$ as follows:
\begin{itemize}
\item With probability $\frac{\sqrt{k}}{k'}$, it chooses the candidate $i$.
\item With probability $1-\frac{\sqrt{k}}{k'}$, it draws a sample from the distribution $D|_I$.
\end{itemize}
By \cref{lem:lottery_on_S} which we prove next, we have
\[
\Pr_{v \sim S} [{\hat D}^{k'} \succ_v D^*] \geq \Pr_{v \sim S} [D^k \succ_v D^*] - \frac{1}{k^2} - 2e^{-\frac{\sqrt{k}}{3}} \geq 1 - \frac{1}{\alpha (k + 1)} - \frac{1}{k^2} - 2e^{-\frac{\sqrt{k}}{3}}.
\]
On the other hand, since $D^*$ is a stable defensive $k'$-lottery on $S$, by definition,
\[
\Pr_{v \sim S} [{\hat D}^{k'} \succ_v D^*] \leq 1 - \frac{1}{k' + 1}.
\]
Putting them together,
\[
1 - \frac{1}{\alpha (k + 1)} - \frac{1}{k^2} - 2e^{-\frac{\sqrt{k}}{3}} \leq 1 - \frac{1}{k' + 1}.
\]
As $\alpha\le 1$, the above inequality implies that
\[
\alpha k \leq k' + 1 \leq (1 - \lambda)k + \Theta( \sqrt{k \log k} ).
\]
The lemma statement follows from rearranging the terms.
\end{hideproofBtwo}

\begin{lemma}
\label{lem:lottery_on_S}
The following inequality holds for $D$, $\hat D$, and any distribution $D'$ over the candidates.
\[
\Pr_{v \sim S} [{\hat D}^{k'} \succ_v D'] \geq \Pr_{v \sim S} [D^k \succ_v D'] - \frac{1}{k^2} - 2e^{-\frac{\sqrt{k}}{3}}.
\]
Recall that (I) $k' = \Big\lceil(1 - \lambda)k + 2\sqrt{k} + \sqrt{k\ln k}\Big\rceil$, (II) $D$ satisfies $p_J = \lambda$, (III) $\hat D$ is a combination of $\{i\}$ with weight $\frac{\sqrt{k}}{k'}$ and $D|_I$ (where $I=C\setminus J$) with weight $1 - \frac{\sqrt{k}}{k'}$, and (IV) $i \succ J$ for voters in $S$.
\end{lemma}

By definition of $\lambda$, the distribution $D$ can be alternatively defined as follows: 
\begin{itemize}
\item With probability $\lambda$, it draws a sample from the distribution $D|_J$.
\item With probability $1-\lambda$, it draws a sample from the distribution $D|_I$.
\end{itemize}

Intuitively, $k$ samples from distribution $D$ consist of $\lambda k$ samples from $D|_J$ and  $(1 - \lambda) k$ samples from $D|_I$. Since voters in $S$ prefer $i$ than all candidates in $J$, the lottery $D^k$ is almost stochastically dominated by ${\hat D}^{k'}$ due to concentration inequalities.

\begin{hideproofBtwo}[Proof of \cref{lem:lottery_on_S}]
Let $A_1$ be the event that ${\hat D}^{k'}$ chooses element $i$ (as opposed to drawing from $D|_I$) at least $1$ and at most $2 \sqrt{k}$ times. Let $A_2$ be the event that $D^k$ draws samples from $D|_I$ (as opposed to $D|_J$) at most $(1 - \lambda) k + \sqrt{k \ln k}$ times. We will next prove that each of the events $A_1$ and $A_2$ happens with probability close to $1$.

\paragraph{For $A_1$:} In expectation, ${\hat D}^{k'}$ chooses element $i$ for $\sqrt{k}$ times. By a multiplicative Chernoff bound,
\[
\Pr[A_1] \geq 1 - 2e^{-\frac{\sqrt{k}}{3}}.
\]

\paragraph{For $A_2$:} In expectation, $D^k$ draws $(1 - \lambda) k$ times from $D|_I$. By an additive Chernoff bound,
\[
\Pr[A_2] \geq 1 - \frac{1}{k^2}.
\]

Note that, conditioned on the event $A_1 \land A_2$,  $D^k$ is first-order stochastically dominated by ${\hat D}^{k'}$ in the preference ordering of every voter in $S$. The lemma holds since $\Pr[A_1 \land A_2] \geq 1 - \frac{1}{k^2} - 2e^{-\frac{\sqrt{k}}{3}}$. 
\end{hideproofBtwo}

%% file: sections/determ_any.tex
\NewEnviron{hideproofBthree}[1][Proof]{%
  \ifshowproofs
    \begin{proof}[#1]
      \BODY
    \end{proof}
  \else
  \par\noindent\textit{#1.} See \Cref{prf:det-k}. \hfill$\square$%
  \par\nobreak\vskip\medskipamount
  \fi
}

\subsection{Deterministic Selection in \texorpdfstring{$k$}{k}-Tournaments: \texorpdfstring{\nameref{box:simul_veto}}{Simultaneous Lottery Veto}}\label{sec:det-k}
We first present the notion of a quasi-kernel, which we will use in \nameref{box:quasi-kernel}. It refers to an independent set in a directed graph that can reach all vertices in at most two hops.
\begin{definition}[quasi-kernel \cite{chvatal1974every}]
\label{def:quasi-kernel}
In any (unweighted) directed graph $G = (V, E)$, a \emph{quasi-kernel} (sometimes called a \emph{semi-kernel}) $K$ is an independent set of vertices such that for any vertex $v \in V \setminus K$, there is a path of length at most $2$ from a vertex $u \in K$ to the vertex $v$.
\end{definition}
Quasi-kernels always exist \cite{chvatal1974every} and are efficiently computable via a simple algorithm \cite{croitoru2015note}.
\begin{theorem}[\cite{chvatal1974every,croitoru2015note}]
Every directed graph $G = (V, E)$ has a quasi-kernel, and one can be computed in $O(|V| + |E|)$ time.
\end{theorem}

Next, we describe the procedure of \nameref{box:quasi-kernel}. It has in effect been used in the voting rule ``$\beta$-Random Dictatorship on the Quasi-Kernel'' from \cite{DBLP:journals/jacm/CharikarRWW24}; here, we distill the procedure and apply it for other purposes.
\begin{mybox}[label={box:quasi-kernel},nameref={quasi-kernel pruning}]{Quasi-Kernel Pruning (with Parameter $\theta \in (\frac{1}{2}, 1]$)}
\begin{itemize}
    \item[$\triangleright$] Build a directed graph whose vertices represent the candidates in $C$.
    \item[$\triangleright$] Draw an edge from $a \in C$ to $b \in C \setminus \{a\}$ if $s_{a \succ b} \geq \theta$.
    \item[$\triangleright$] Compute a quasi-kernel in the graph and discard all vertices not in the quasi-kernel.
\end{itemize}
\end{mybox}
\cref{lem:pruning_loss} bounds the ``loss'' of \nameref{box:quasi-kernel}.
\begin{lemma}
\label{lem:pruning_loss}
After running \nameref{box:quasi-kernel}, there is always an unpruned candidate whose social cost is at most $\frac{4}{\theta} - 3$ times the optimal social cost.
\end{lemma}
\begin{proof}
\cref{cor:two-step} and the two-hop property of quasi-kernels imply the lemma statement.
\end{proof}

Our voting rule requires the following definition.
\begin{definition}[reverse stable lotteries]
For a set $C$ of candidates and an ordering $\sigma$ on $C$, the \emph{reverse ordering} of $\sigma$ is an ordering $\sigma^R$ such that $c\succ_{\sigma} c'$ if and only if $c'\succ_{\sigma^R} c$. Given a preference profile, its \emph{reverse stable $k$-lottery} is the stable $k$-lottery on the profile in which the preference ordering of each voter is replaced by its reverse ordering.
\end{definition}

Now we are ready to present \nameref{box:simul_veto}. It is inspired by the Simultaneous Veto rule of \cite{DBLP:conf/sigecom/Kizilkaya023}.
\begin{mybox}[label={box:simul_veto},nameref={Simultaneous Lottery Veto}]{The \underline{Simultaneous Lottery Veto} Voting Rule (with Parameters $k, \theta$)}
\begin{itemize}
\item[$\triangleright$] Run \nameref{box:quasi-kernel} with parameter $\theta$. Let $\hat C$ be the remaining candidates.

\item[$\triangleright$] Compute a stable $k$-lottery on $\hat C$. For each $c\in \hat C$, initialize $\score(c)$ as the probability of picking $c$ in the stable $k$-lottery.

\item[$\triangleright$] At each time $t \in [0, 1]$, compute a reverse stable $k$-lottery $\Delta_t$ on the set of candidates $\{c \in \hat C \mid \score(c) > 0\}$. Decrease the vector $\score(\cdot)$ at rate $\Delta_t$.

\item[$\triangleright$] The candidate $c$ with positive $\score(c)$ for the longest time wins. (Break ties arbitrarily.)
\end{itemize}
\end{mybox}

\cite{DBLP:conf/sigecom/Kizilkaya023} proves that Simultaneous Veto has distortion $3$, which is optimal among deterministic ranked voting rules. \cite{DBLP:journals/jacm/CharikarRWW24} has observed that the similar rule of Plurality Veto \cite{DBLP:conf/ijcai/KizilkayaK22} can be proved within the biased metric framework. In fact, the biased metric framework admits simple analysis for Simultaneous Veto and \nameref{box:simul_veto} as well.

We start our proof with the following observation.
\begin{observation}
\label{obs: score}
For every subset $J\subseteq C$, $\int_0^1 \Delta_t(J) \d t = \score(J)$.
\end{observation}
\begin{hideproofBthree}[Proof of \Cref{obs: score}]
Initially, since $\score(\cdot)$ is a distribution, $\sum_c\score(c)=1$. At any time $t$, the rate that the value $\sum_c\score(c)$ is being decreased is $\sum_c \Delta_t(c)=1$. Therefore, after $1$ unit of time, the value of $\sum_c\score(c)$ becomes $0$ and so $\score(\cdot)$ becomes an all-$0$ vector, and therefore for each subset $J\subseteq C$, the total amount that the coordinates $\{\score(c)\}_{c\in J}$ have been decreased is $\int_0^1 \Delta_t(J) \d t = \score(J)$.
\end{hideproofBthree}

\begin{theorem}
\label{thm:distortion_simul_veto}
The distortion of \nameref{box:simul_veto} with parameter $\theta = \frac{1}{1 + \sqrt{g(k)}}$ (where $g(k)$ is  as in \nameref{lem:main}) is at most $3 + O\left(\left(\frac{\log k}{k}\right)^{\frac{1}{4}}\right)$.
\end{theorem}
\begin{hideproofBthree}[Proof of \Cref{thm:distortion_simul_veto}]
Let $j^*$ be the winning candidate of \nameref{box:simul_veto}. Let $I \sqcup J$ be any partition of $C$ with $j^*\in J$. We define $\score(J):=\sum_{c\in J}\score(c)$ and $\Delta_t(J):=\sum_{c\in J}\Delta_t(c)$, and define $\score(I),\Delta_t(I)$ similarly.
By definition, $\score(I)+\score(J)=1$.

From \nameref{lem:main} (\cref{lem:main}), for each $i\in I$,
\[
\score(J) - g(k) \leq 1 - s_{i \succ J}.
\]
Applying \nameref{lem:main} to the reversed preferences, we get that
\[
\Delta_t(I) - g(k) \leq 1 - s_{I\succ j^*}, \qquad \forall t \in (0, 1).
\]
From \Cref{obs: score}, for each $i\in I$,
\[
\begin{split}
s_{I \succ j^*} & \leq \int_0^1 \bigg(1+g(k) - \Delta_t(I)\bigg)\d t \leq 
1+g(k)-\score(I)\\
& \leq 
g(k)+\score(J)\le g(k)+\bigg(1 - s_{i \succ J} +g(k)\bigg)=2g(k)+1-s_{i \succ J}.
\end{split}
\]
From the \nameref{box:quasi-kernel} process, for each pair $i,j$ of candidates in the set $\hat C$ it produces, $s_{i \succ j}\le \theta$. Therefore, $s_{i \succ J} \leq \theta$, so $1 - s_{i \succ J} \geq 1 - \theta$, and so
\[
s_{I \succ j^*} \leq 2g(k)\cdot \frac{1-s_{i \succ J}}{1 - \theta} + (1- s_{i \succ J}) = \left(\frac{2g(k)}{1 - \theta} + 1\right) \cdot (1 - s_{i \succ J}).
\]
From \Cref{cor:lp-condition} and \Cref{lem:pruning_loss}, the distortion of \nameref{box:simul_veto} is at most
\[
\left(1 + 2 \cdot \left(\frac{2g(k)}{1 - \theta} + 1\right)\right) \cdot \left(\frac{4}{\theta} - 3\right).
\]
For $\theta = \frac{1}{1 + \sqrt{g(k)}}$, %
this simplifies to
\[
\left(3 + 4\sqrt{g(k)} + 4g(k)\right) \cdot \left(1 + 4\sqrt{g(k)}\right) = 3 + O\left(\sqrt{g(k)}\right) = 3 + O\left(\left(\frac{\log k}{k}\right)^{\frac{1}{4}}\right). \qedhere
\]
\end{hideproofBthree}

%% file: sections/random_any.tex
\NewEnviron{hideproofBfour}[1][Proof]{%
  \ifshowproofs
    \begin{proof}[#1]
      \BODY
    \end{proof}
  \else
    \par\noindent\textit{#1.} See \Cref{prf:random-any}. \hfill$\square$%
  \par\nobreak\vskip\medskipamount
  \fi
}

\subsection{Randomized Selection in \texorpdfstring{$k$}{k}-Tournaments: \texorpdfstring{\nameref{box:two_lotteries}}{Pruned Double Lotteries}}\label{sec:random-any}
\begin{mybox}[label={box:two_lotteries},nameref={Pruned Double Lotteries}]{The \underline{Pruned Double Lotteries} Voting Rule (with Parameters $k, \mu, \theta$)}
\begin{itemize}
\item[$\triangleright$] With probability $\mu$, select a candidate according to a stable $1$-lottery.

\item[$\triangleright$] With probability $1 - \mu$, run \nameref{box:quasi-kernel} with parameter $\theta$. Select a candidate according to a stable $k$-lottery on the remaining candidates.
\end{itemize}
\end{mybox}

\begin{theorem}
\label{thm:distortion_two_lotteries}
There is a universal constant $r < 3$ and universal parameters $\mu$ (close to $1$) and $\theta$ (slightly larger than $0.5$), so that the distortion of \nameref{box:two_lotteries} with these parameters is at most $r$ for any $k \geq 2$. (It is a $(k+1)$-tournament rule.)
\end{theorem}

It is implicit in \cite{DBLP:journals/jacm/CharikarRWW24} that if a voting rule has distortion less than $3$ among the remaining candidates after \nameref{box:quasi-kernel}, then it can have distortion less than $3$ in general after \nameref{box:quasi-kernel} and mixing with a stable $1$-lottery (Maximal Lotteries).  The next lemma rephrases this fact in the context of \nameref{box:two_lotteries}.
\begin{lemma}[implicit in \cite{DBLP:journals/jacm/CharikarRWW24}]
\label{lem:reduce_to_balanced}
We call a preference profile $\theta$-regular if $s_{i \succ j} \leq \theta$ for all $i, j \in C$. Fix a parameter $\theta \in (0.5, 1]$.
If the distortion constant of stable $k$-lotteries is less than $r < 3$ (independent of $k$) among $\theta$-regular profiles for all $k \geq 2$, then the distortion constant of \nameref{box:two_lotteries} with a sufficiently large parameter $\mu < 1$ (independent of $k$) is less than $r' < 3$ (independent of $k$) in general for all $k \geq 2$.
\end{lemma}
\begin{hideproofBfour}[Proof sketch of \Cref{lem:reduce_to_balanced}]
\cite{DBLP:conf/soda/CharikarR22} showed that to bound distortion, without loss of generality, we only need to consider biased metrics.

\cite{DBLP:journals/jacm/CharikarRWW24} defined the notion of $(\alpha, \beta)$-consistency, which says that if under an $(\alpha, \beta)$-consistent biased metric, a candidate $j$ satisfies $s_{j \succ i^*} \geq \beta$ where $i^*$ is the optimal candidate, then $j$ must be close (parametrized by $\alpha$) to the optimal candidate.
They also
showed that stable $1$-lotteries (Maximal Lotteries) have distortion $3$, and the distortion becomes better than $3$ as the biased metric becomes more inconsistent.

Any deterministic selection after \nameref{box:quasi-kernel} with parameter $\theta \in [0.5, 0.99]$ always has bounded distortion, because in the tournament graph there is a short path from the selected candidate to the optimal candidate with edge weights bounded away from zero (see e.g., \cref{cor:two-step}).

Therefore, we randomize between (I) a stable $1$-lottery, and (II) a stable $k$-lottery after \nameref{box:quasi-kernel}. As long as we can show that the stable $k$-lottery has distortion better than $3$ by a constant for sufficiently consistent biased metrics, then the overall rule of \nameref{box:two_lotteries} has distortion better than $3$ for a sufficiently large $\mu$:
\begin{itemize}
\item If the biased metric is not sufficient consistent, then rule (I) achieves distortion less than $3$ by a constant and rule (II) has bounded distortion. Therefore, randomizing between them with a sufficiently large $\mu$ achieves distortion less than $3$ by a constant.
\item On the other hand, if the biased metric is sufficiently consistent, then rule (I) achieves distortion at most $3$ and rule (II) has distortion less than $3$ by a constant (since by $(\alpha, \beta)$-consistency, there is a candidate in the quasi-kernel being very close to the optimal candidate, and hence \nameref{box:quasi-kernel} incurs little loss). Randomizing between them achieves distortion less than $3$ by a constant.
\end{itemize}
This concludes the proof sketch of the lemma. (To see how these ideas can be formally executed, we refer the reader to \cite{DBLP:journals/jacm/CharikarRWW24}.)
\end{hideproofBfour}

The next lemma has a similar form as \nameref{lem:main}, but is more useful for a small $k$.
\begin{lemma}
\label{lem:small_stability_representation}
Fix a stable $k$-lottery $D$ and let $p_J$ be its probability of selecting a candidate in $J$. It holds that
\[
s_{i \succ J} \leq \frac{p_J^{-k}}{k + 1}.
\]
\end{lemma}
\begin{hideproofBfour}[Proof of \Cref{lem:small_stability_representation}]
Since $D$ is a stable $k$-lottery, we have
$\Pr_{v \sim V}[i \succ_v D^k] \leq \frac{1}{k + 1}$. On the other hand, let $S$ be the subset of voters with the preference $i \succ J$ and we have
\begin{align*}
\Pr_{v \sim V}[i \succ_v D^k] = &s_{i \succ J} \cdot \Pr_{v \sim S}[i \succ_v D^k] + (1 - s_{i \succ J}) \cdot \Pr_{v \sim C \setminus S}[i \succ_v D^k]\\
\geq &s_{i \succ J} \cdot \Pr_{v \sim S}[i \succ_v D^k] \geq s_{i \succ J} \cdot p_J^k.
\end{align*}
Therefore, $s_{i \succ J} \cdot p_J^k \leq \frac{1}{k + 1}$.
\end{hideproofBfour}

We are finally ready to prove \cref{thm:distortion_two_lotteries}.
\begin{hideproofBfour}[Proof of \cref{thm:distortion_two_lotteries}]
Using \cref{lem:reduce_to_balanced} and \cref{cor:lp-condition}, we wish to show that in the case of $s_{i \succ j} \leq \theta$ for all $i, j \in C$, it holds that
\[
\sum_{j \in J} s_{I \succ j} p_j \leq \lambda (1 - s_{i \succ J})
\]
with $\lambda < 1$. Writing $p_J = \sum_{j \in J} p_j$, since $s_{I \succ j} \leq \theta$, it is sufficient to show
\begin{equation}
\label{eq:lotteries_to_show}
p_J \leq \frac{\lambda}{\theta} \cdot (1 - s_{i \succ J}).
\end{equation}

\paragraph{Small $k$:}
If $p_J \leq \frac{\lambda}{\theta} \cdot (1 - \theta)$, then \cref{eq:lotteries_to_show} automatically holds. Using \cref{lem:small_stability_representation}, we have
\[
\frac{\lambda}{\theta} \cdot (1 - s_{i \succ J}) \geq \frac{\lambda}{\theta} \cdot \left(1 - \frac{p_J^{-k}}{k + 1} \right).
\]
Therefore, if $p_J \leq \frac{\lambda}{\theta} \cdot \left(1 - \frac{p_J^{-k}}{k + 1}\right)$, then \cref{eq:lotteries_to_show} also holds.

Putting them together, we know that is enough to cover all $p_J \in [0, 1]$ if we set
\[
\lambda = \frac{\theta}{1 - \theta} \cdot \left(\frac{1}{\theta(k + 1)}\right)^\frac{1}{k}.
\]
For any finite $k \geq 2$, this $\lambda$ is strictly less than $1$ when $\theta > 0.5$ is sufficiently close to $0.5$.

\paragraph{Large $k$:}
Using \nameref{lem:main}, we know that
\[
p_J \leq g(k) + 1 - s_{i \succ J}.
\]
Since $1 - s_{i \succ J} \geq 1 - \theta$, we further have
\[
p_J \leq \left(\frac{g(k)}{1 - \theta} + 1\right) \cdot (1 - s_{i \succ J}).
\]
Therefore, setting $\lambda = \theta + \frac{\theta}{1 - \theta} \cdot g(k)$ is sufficient for \cref{eq:lotteries_to_show}. This is less than $1$ for sufficiently large $k$ for any $\theta < 1$.

\paragraph{Summary:} The discussion above shows that there exist parameters $r < 3$ and $\theta > 0.5$, so that for all $k \geq 2$, stable $k$-lotteries have distortion less than $r$ among profiles in which $s_{i \succ j} \leq \theta$ for all $i, j \in C$. Invoking \cref{lem:reduce_to_balanced} completes the proof.
\end{hideproofBfour}

%% file: sections/conclusions.tex
\section{Conclusions}\label{sec:conclusions}

We conclude with a retrospective discussion of our results and suggestions for future work.

\paragraph{A Condorcet view of distortion $3$.} Part of the reason why $3$ is a natural goal (and barrier) for many types of voting rules is that it is the distortion achieved by a \emph{Condorcet winner}---a candidate preferred to any other by a majority of voters. In particular when a majority of voters prefers $j^*$ over $i^*$, this implies that $\SC(j^*) \leq 3\SC(i^*)$. Since Condorcet winners do not always exist (a fact known as \emph{Condorcet's paradox}), a distortion $3$ candidate can be viewed as a relaxation of a Condorcet winner that \emph{does} always exist.  

Though a Condorcet winner (when one exists) can be identified using the tournament graph, \Cref{thm:lb} shows that it is not always enough to identify a distortion 3 candidate. One can view this result as a strengthening of Condorcet's paradox in the metric distortion setting. On the other hand, we also show that with the additional information in $k$-tournaments, we can get distortion approaching $3$ with deterministic rules (\Cref{thm:distortion_simul_veto}), and even beyond 3 with randomization (\Cref{thm:distortion_two_lotteries}).

\paragraph{Closing metric distortion gaps.} A natural question arising from our work is whether the gap of $(3.1128, 3.9312)$ for optimal distortion of deterministic tournament rules can be closed. We anticipate this may be challenging. As with metric distortion in randomized voting rules, a tight upper bound would likely need to show that a highly asymmetric and unnatural instance is the worst-case. More specifically, we note some specific challenges to improving our approach in \Cref{rem:ub-hard,rem:lb-hard}. We find it somewhat intriguing that in \Cref{rem:ub-hard}, a wishful argument for analyzing the biased metrics actually suggests the Ranked Pairs rule. Even though fully closing the gap may be challenging, we are optimistic that natural deterministic tournament rules (such as an appropriate modification of Ranked Pairs) may improve the upper bound.

We are also optimistic that some of our techniques may be useful in resolving other problems in metric distortion, particularly the optimal metric distortion of randomized voting rules, where the optimal metric distortion is currently known to be in $(2.112, 2.753)$ \cite{DBLP:conf/soda/CharikarR22,DBLP:journals/jacm/CharikarRWW24}. As noted in the introduction, \cite{DBLP:journals/jacm/CharikarRWW24} suggested that a deterministic tournament rule with distortion less than $2 + \sqrt{5}$ could be a useful ingredient for improved randomized voting rules, and it would be interesting to explore whether \nameref{box:unblanketed} or variants of it could play this role. On the lower bound front, we find it intriguing that our lower bound mostly uses the same $(0,1,2,3)$-metrics used in \cite{DBLP:conf/soda/CharikarR22}'s lower bound, but also a ``half-integral'' version of these metrics for one of the cases. Though \cite{DBLP:conf/soda/CharikarR22} conjectured that their lower bound is optimal, we suspect that their bound can be improved by considering more sophisticated metrics such as these.

\paragraph{A theory of $k$-tournament rules.} One of the conceptual contributions of our work is the generalization of tournament rules to $k$-tournament rules, and the introduction of natural $k$-tournament rules inspired by algorithms for committee selection. Considering the vastness of the literature on tournament rules, we see this as a potential treasure trove of untapped voting rules that may be practically useful. We leave this as an exciting direction for future work.

%% file: sections/omitted.tex
\section{Relevant External Facts}

We show that \Cref{lem:lp-tournament} implies the result due to \cite{DBLP:conf/ec/MunagalaW19} that in an election with a cyclically symmetric tournament graph, every candidate has distortion at most 3.
\begin{theorem}[\cite{DBLP:conf/ec/MunagalaW19}]\label{thm:cyc-sym}
Suppose there exists a cyclic permutation $\tau$ on the candidates such that for all candidates $i$ and $j$, $s_{i\cg j} = s_{\tau(i)\cg \tau(j)}$. Then for all candidates $j^*$ and $i^*$, $\SC(j^*) \leq 3 \SC(i^*)$.
\end{theorem}

\begin{proof}
Without loss of generality, we can assume that $\tau(j^*) = i^*$, since we can replace $\tau$ with some $\tau^t$ so that this is the case. 

Suppose we have an arbitrary partition $I \sqcup J = C$ such that $i^* \in I,  j^* \in J$. To apply \Cref{lem:lp-tournament}, it suffices to show that 
$$\min_{i\in I}s_{i\cg j^*} \leq \max_{j\in J} s_{j \cg i^*}.$$

Let $t$ be such that $\tau^t(j^*) \in I$ but $\tau^{t + 1}(j^*) \in J$. Such a $t$ must exist since $\tau^1(j^*) = i^* \in I$ and $\tau^m(j^*) = j^* \in J$. 

Then observe that
$$\min_{i\in I}s_{i\cg j^*} \leq s_{\tau^{t}(j^*) \cg j^*} = s_{\tau^{t + 1}(j^*) \cg i^*} \leq \max_{j\in J} s_{j \cg i^*}$$
as desired.
\end{proof}

Next, we revisit the \emph{Ranked Pairs} rule, due to Tideman \cite{tideman1987independence}.  \cite{DBLP:journals/ai/AnshelevichBEPS18} intriguingly showed that this voting rule has distortion at most 3 when there are at most 4 candidates, leading them to conjecture that it has distortion at most 3 in general. We will give a simple proof of this fact, leveraging \Cref{lem:post-shift}.

\begin{mybox}[label={box:ranked_pairs},nameref={Ranked Pairs}]{The \underline{Ranked Pairs} Voting Rule.}
\begin{itemize}
\item[$\triangleright$] Initialize a directed graph $G$ whose vertices are the candidates, and no edges.

\item[$\triangleright$] Iterating over each ordered pair $(i, j)$ of candidates in decreasing order of $s_{i \cg j}$:

\begin{itemize}
    \item Add the edge $(i, j)$ to $G$ if it does not create a cycle. 
\end{itemize}

\item[$\triangleright$] Return the candidate with no in-edges in $G$.

\end{itemize}
\end{mybox}

\begin{theorem}[\cite{DBLP:journals/ai/AnshelevichBEPS18}]\label{thm:rp}
In elections with at most 4 candidates, ranked pairs has distortion at most 3.
\end{theorem}

\begin{proof}
The key property of Ranked Pairs is that if the winner is $j^*$, then for all candidates $i^* \neq j^*$, there exists a path $j^*, k_1, k_2, \ldots, k_r, i^*$ such that 
\begin{equation}\label{eq:path}
s_{i^* \cg j^*} \leq s_{j^* \cg k_1}, s_{k_1 \cg k_2}, \dots, s_{k_r \cg i^*}.  
\end{equation}
(Otherwise, the edge $(i^*, j^*)$ would have been added to $G$.) If there are at most 4 candidates, then this path has length at most 3. Let us suppose that the path is $j^*, k_1, k_2, i^*$, where we allow that $k_1 = k_2$ if the path is length 2, and $k_1 = k_2 = i^*$ if the path has length 1.

We will verify that the conditions of \Cref{lem:post-shift} are satisfied with $\lambda = 1$ and $k\leftarrow k_2$. \cref{eq:path} gives us that either $k_2 = i^*$ or $s_{i^* \cg j^*} \leq s_{k_2 \cg i^*}.$ Now suppose that we have an arbitrary partition of candidates $I\sqcup J$ such that $i^*, k_2 \in I$ and $j^* \in J$. If $k_1 \in I$ then we have
$$\min_{i\in I} s_{i \cg j^*} \leq s_{k_1 \cg j^*} \leq s_{j^* \cg i^*} \leq \max_{j\in J} s_{j \cg i^*}.$$
On the other hand, if $k_1 \in J$ (which means $k_1 \neq k_2$) then we have
$$\min_{i\in I} s_{i \cg j^*} \leq s_{i^* \cg j^*} \leq s_{k_1 \cg k_2} \leq \max_{j\in J} s_{j \cg k_2}.$$
Therefore, the conditions of \Cref{lem:post-shift} are satisfied, and it follows that $\SC(j^*) \leq 3 \SC(i^*).$
\end{proof}